\documentclass[prx,a4paper,aps,onecolumn,superscriptaddress,11pt,accepted=2017-12-28,longbibliography]{quantumarticle}

\newcommand{\ket}[1]{|#1\rangle}
\newcommand{\bra}[1]{\langle#1|}
\newcommand{\braket}[2]{\langle#1|#2\rangle}

\usepackage[utf8]{inputenc}
\usepackage[english]{babel}
\usepackage[T1]{fontenc}
\usepackage{amsmath, amssymb, amsthm}
\usepackage{hyperref}
\usepackage{bbm}

\usepackage{tikz}
\usepackage{lipsum}

\newtheorem{theorem}{Theorem}[section]
\newtheorem{lemma}{Lemma}[section]
\newtheorem{defn}{Definition}[section]
\newtheorem{cor}{Corollary}[section]

\usepackage[numbers,sort&compress]{natbib}

\begin{document}

\title{Separability of diagonal symmetric states: a quadratic conic optimization problem}
\date{\today}
\author{Jordi Tura}
\affiliation{ICFO - Institut de Ciencies Fotoniques, The Barcelona Institute of Science and Technology, 08860 Castelldefels (Barcelona), Spain}
\affiliation{Max-Planck-Institut f\"ur Quantenoptik, Hans-Kopfermann-Stra{\ss}e 1, 85748 Garching, Germany}
\email{jordi.tura@mpq.mpg.de}
\author{Albert Aloy}
\affiliation{ICFO - Institut de Ciencies Fotoniques, The Barcelona Institute of Science and Technology, 08860 Castelldefels (Barcelona), Spain}
\author{Rub\'en Quesada}
\affiliation{Departament de F\'isica, Universitat Aut\`onoma de Barcelona, E-08193 Bellaterra, Spain}
\author{Maciej Lewenstein}
\affiliation{ICFO - Institut de Ciencies Fotoniques, The Barcelona Institute of Science and Technology, 08860 Castelldefels (Barcelona), Spain}
\affiliation{ICREA, Pg. Lluís Companys 23, E-08010 Barcelona, Spain}
\author{Anna Sanpera}
\affiliation{Departament de F\'isica, Universitat Aut\`onoma de Barcelona, E-08193 Bellaterra, Spain}
\affiliation{ICREA, Pg. Lluís Companys 23, E-08010 Barcelona, Spain}

\maketitle

\begin{abstract}
We study the separability problem in mixtures of Dicke states \textit{i.e.}, the separability of the so-called Diagonal Symmetric (DS) states. First, we show that separability in the case of DS in ${\mathbbm C}^d\otimes {\mathbbm C}^d$ (symmetric qudits) can be reformulated as a quadratic conic optimization problem. This connection allows us to exchange concepts and ideas between quantum information and this field of mathematics. For instance, copositive matrices can be understood as indecomposable entanglement witnesses for DS states. As a consequence, we show that positivity of the partial transposition (PPT) is sufficient and necessary for separability of DS states for $d \leq 4$. Furthermore, for $d \geq 5$, we provide analytic examples of PPT-entangled states. Second, we develop new sufficient separability conditions beyond the PPT criterion for bipartite DS states. Finally, we focus on $N$-partite DS qubits, where PPT is known to be necessary and sufficient for separability. In this case, we present a family of almost DS states that are PPT with respect to each partition but nevertheless entangled.
\end{abstract}

\section{Introduction}

Entanglement \cite{HorodeckiRMP} is one of the most striking features of quantum physics, departing entirely from any classical analogy. Furthermore, entanglement is a key resource for quantum information processing tasks, such as quantum cryptography \cite{Ekert91} or metrology \cite{gross2010nonlinear}. Importantly, entanglement is a necessary resource to enable the existence of Bell correlations \cite{bell1964einstein, WernerState}, which are the resource device-independent quantum information processing is built upon \cite{DIQIPAcin}. Despite its both fundamental and applied interest, the so-called separability problem (\textit{i.e.}, deciding whether a quantum state is entangled or not, given its description) remains open except for very specific cases. Although this problem has been shown to be, in the general case, NP-hard \cite{Gurvits03}, it remains unclear whether this is also the case for physical systems of interest, where symmetries appear in a natural way.

To tackle the separability problem, simple tests have been put forward, which give a partial characterization of entanglement. The most celebrated entanglement detection criterion is the so-called positivity under partial transposition (PPT) criterion \cite{Peres}. It states that every state that is not entangled must satisfy the PPT criterion. Therefore, states that break the PPT criterion are entangled. Unfortunately, the converse is true only in very low-dimensional systems \cite{HHH96}, such as two qubit \cite{stormer1963} or qubit-qutrit systems \cite{WORONOWICZ1976165}. Examples of entangled states satisfying the PPT criterion have been found for strictly larger-dimensional systems \cite{Horodecki3x3-2x4}.

Symmetries are ubiquitous in Nature and they play a fundamental role in finding an efficient description of physical systems. The so-called symmetric states constitute an important class of quantum systems to describe systems of indistinguishable particles \cite{Eckert02}. Symmetric states can be mapped to spin systems that are invariant under the exchange of particles and, moreover, they are spanned solely by the largest-spin subspace in the Schur-Weyl duality representation \cite{goodman2009symmetry}. The Dicke states \cite{Dicke54} provide a convenient basis to represent symmetric states. Moreover, Dicke states are also experimentally available \cite{Dicke3K, Dicke6Photon, lu2007experimental} and they also appear naturally as ground states of physically relevant Hamiltonians, such as the isotropic Lipkin-Meshkov-Glick model \cite{LMG-JILA}. 
Much theoretical study has been devoted to the characterization of entanglement in qubit symmetric states: $3$-qubit symmetric states are separable if, and only if, they satisfy the PPT criterion \cite{Eckert02}, but this is no longer the case already for $N\geq 4$ \cite{TuraPRA12, AugusiakPRA12}. Despite diverse separability criteria exist for symmetric states (see e.g. \cite{PhDTura}), the separability problem remains still open.

Mixtures of Dicke states are symmetric states that are diagonal in the Dicke basis. These constitute an important class of quantum states which naturally arise e.g. in dissipative systems such as photonic or plasmonic one-dimensional waveguides \cite{GonzalezTudela-PorrasPRL}. Mixtures of Dicke states form a small subclass of the symmetric states. They are the so-called Diagonal Symmetric (DS) states. In this context, the separability problem has also gained interest. For instance, the best separable approximation (BSA, \cite{LewensteinSanperaSep}) has been found analytically for DS states for $N$-qubits \cite{QuesadaBSA}. In \cite{WolfePRL14}, it was conjectured that $N$-qubit DS states are separable if, and only if, they satisfy the PPT criterion with respect to every bipartition.
The conjecture was proven by N. Yu in \cite{KoreanGuy} where, moreover, he observed that PPT is a sufficient and necessary condition for bipartite DS states of qudits with dimension $3$ and $4$, but becomes NP-hard for larger dimensions. Within the $N$-qubit DS set it has been shown in \cite{QuesadaRanaSanpera} that there is a family of states that violate the weak Peres conjecture \cite{Peres1999}:
those states are PPT-bound entangled with respect to one partition, but they violate a family of permutationally invariant two-body Bell inequalities \cite{SciencePaper, AnnPhys, ThetaBodies}.

In experiments, PPT-entangled states have also been recently observed. In the multipartite case, the Smolin state has been prepared with four photons, using the polarization degree of freedom for the qubit encoding \cite{BoundEntanglement4photons, BoundEntanglement4photonsBis}. Very recently, although bound entanglement is the hardest to detect \cite{BaeBoundEnt}, the Leiden-Vienna collaboration has reported the observation of bound entanglement in the bipartite case with two twisted photons, combining ideas of complementarity \cite{ComplementarityNJP} and Mutually Unbiased Bases (MUBs) \cite{SpenglerMUB}.

Here, we independently recover the results of N. Yu \cite{KoreanGuy} by reformulating the problem in terms of optimization in the cone of completely-positive\footnote{
Throughout this paper, the term \textit{completely-positive} corresponds to the definition given in Def. \ref{def:CP} and it is not to be confounded with the concept of a completely positive map that arises typically in a quantum information context.} matrices. First, we revisit the problem of determining separability of two DS qudits in arbitrary dimensions. We show that it can be reformulated in terms of a quadratic conic optimization problem \cite{berman2015open}. In particular, we show that separability in DS states is equivalent to the membership problem in the set of completely-positive matrices. The equivalence between these two problems allows us to import/export ideas between entanglement theory and non-convex quadratic optimization\footnote{
Quadratic conic optimization problems appear naturally in many situations (see \cite{berman2015open} and references therein). These include economic modelling \cite{gray1980nonnegative}, block designs \cite{HornMatrix}, maximin efficiency-robust tests \cite{bose1995maximin}, even Markovian models of DNA evolution \cite{berman2003completely}. Recently, they have found their application in data mining and clustering \cite{ding2005equivalence}, as well as in dynamical systems and control \cite{berman2012characterisation, mason2007linear}.}.
Second, we provide examples of entangled PPTDS states and entanglement witnesses detecting them. Third, we give further characterization criteria for separability in DS states in terms of the best diagonal dominant decomposition. Finally, we present a family of $N$-qubit almost-DS states that are PPT with respect to each bipartition, but nevertheless entangled. The word \textit{almost} here means that by adding an arbitrarily small off-diagonal term (GHZ coherence) to a family of separable DS $N$-qubits, the state becomes PPT-entangled.
\newpage	

The paper is organized as follows. In Section \ref{sec:Preliminaries} we establish the notation and the basic definitions that we are going to use in the next sections. In Section \ref{sec:Separability} we discuss the separability problem for bipartite DS states of arbitrary dimension, with particular emphasis in their connection to non-convex quadratic optimization problems. In Section \ref{sec:sufficient} we provide sufficient criteria to certify either separability or entanglement. In Section \ref{sec:Class} we present a class of PPT-entangled multipartite qubit almost-diagonal symmetric states. In Section \ref{sec:Conclusions} we conclude and discuss further research directions. Finally, in the Appendix we present some proofs, examples and counterexamples that complement the results discussed in the text.

\section{Preliminaries}
\label{sec:Preliminaries}
In this section we set the notation and define the basic concepts that we are going to use throughout the paper.

\subsection{The separability problem}
\begin{defn}

Consider a bipartite quantum state $\rho$ acting on $\mathbbm{C}^d \otimes \mathbbm{C}^{d'}$. The state $\rho$ is positive semi-definite ($\rho \succeq 0$) and normalized ($\mathrm{Tr} \rho = 1$). A state $\rho$ is \textit{separable} if it can be written as
\begin{equation}
\rho = \sum_{i} p_i \rho_i^{A} \otimes \rho_i^{B},
\label{eq:sep}
\end{equation}
where $p_i$ form a convex combination ($p_i \geq 0$ and $\sum\limits_i p_i=1$) and $\rho_i^{A}$ ($\rho_i^{B}$) are quantum states acting on Alice's (Bob's) subsystem; \textit{i.e.}, they are positive semidefinite operators of trace one. If a decomposition of $\rho$ of the form of Eq. (\ref{eq:sep}) does not exist, then $\rho$ is \textit{entangled}.
\end{defn}

The separability problem; i.e., deciding whether a quantum state $\rho$ admits a decomposition of the form of Eq. (\ref{eq:sep}) is, in general, an NP-hard problem \cite{Gurvits03}. However, there exist simple tests that provide sufficient conditions to certify that $\rho$ is entangled \cite{HorodeckiRMP}. One of the most renowned separability criteria is the positivity under partial transposition (PPT) criterion \cite{Peres}. It states that, if $\rho$ can be decomposed into the form of Eq. (\ref{eq:sep}), then the state $(\mathbbm{1} \otimes T)[\rho]$ must be positive semi-definite, where $T$ is the transposition with respect to the canonical basis of $\mathbbm{C}^{d'}$. Such state is denoted $\rho^{T_B}$, the partial transposition of $\rho$ on Bob's side. Because $(\rho^{T_B})^T = \rho^{T_A}$, the PPT criterion does not depend on which side of the bipartite system the transposition operation is applied on. Breaking PPT criterion is a necessary and sufficient condition for entanglement only in the two qubit \cite{stormer1963} and qubit-qutrit \cite{WORONOWICZ1976165} cases, and there exist counterexamples for states of strictly higher physical dimension \cite{Horodecki3x3-2x4}.

In the multipartite case, the definition of separability given in Eq. (\ref{eq:sep}) naturally generalizes to $N$ subsystems.

\begin{defn}
A quantum state $\rho$ acting on $\mathbbm{C}^{d_1}\otimes \cdots \otimes \mathbbm{C}^{d_N}$ is \textit{fully separable} if it can be written as
\begin{equation}
\rho = \sum_{i} p_i \rho_i^{(A_1)} \otimes \cdots \otimes \rho_i^{(A_N)},
\label{eq:fullysep}
\end{equation}
where $\rho_i^{(A_k)}$ are quantum states acting on the $k$-th subsystem and $p_i$ form a convex combination.
\end{defn}
Therefore, the PPT criterion also generalizes to $2^{\lfloor N/2\rfloor}$ criteria, where $\lfloor \cdot \rfloor$ is the floor function, depending on which subsystems one chooses to transpose.

\subsubsection{Entanglement witnesses}
Let us denote by $\cal{D}_{\mathrm{sep}}$ the set of separable sates (cf. Eqs. (\ref{eq:sep}), (\ref{eq:fullysep})). This set is closed and convex. Therefore it admits a dual description in terms of its dual cone, which we denote $${\cal P} = \{ W = W^\dagger \text{ s. t. } \langle W, \rho \rangle \geq 0 \ \forall \rho \in {\cal D}_{\mathrm{sep}}\},$$
where the usual Hilbert-Schmidt scalar product $\langle W, \rho \rangle = \mathrm{Tr} (W^\dagger \rho)$ is taken. The elements of ${\cal P}$ can be thus viewed as half-spaces containing ${\cal D}_{\mathrm{sep}}$. Of course, not every operator in $\cal P$ is useful to detect entangled states. In order to be non-trivial, one requires that $W$ has at least one negative eigenvalue. Such operators are called \textit{entanglement witnesses} (EW) \cite{TerhalEW99} and they form a non-convex set, denoted ${\cal W} = \{W \in {\cal P} \text{ s. t. } W \not \succeq 0\}$. A state $\rho$ is then separable if, and only if, $\mathrm{Tr} (W \rho) \geq 0$ for all $W \in {\cal W}$.

Among EWs, it is worth to make a distinction that relates them to the PPT criterion: decomposable and indecomposable EWs.
\begin{defn}
 \textit{Decomposable} EWs (DEWs) in a bipartite quantum system are those $W\in {\cal W}$ of the form 
\begin{equation}
W = P + Q^{T_B},
\label{eq:DEW}
\end{equation}
with $P \succeq 0$ and $Q \succeq 0$. \textit{Indecomposable} EWs (IEWs) are those EWs that are not of the form of Eq. (\ref{eq:DEW}).
\end{defn}
Although DEWs are easier to characterize \cite{SummerFellow}, they do not detect PPT-entangled states, because
\begin{equation}
\mathrm{Tr}(W\rho) = \mathrm{Tr}(P\rho)  +\mathrm{Tr} (Q^{T_B}\rho) = \mathrm{Tr}(P\rho)  +\mathrm{Tr} (Q\rho^{T_B}) \geq 0.
\end{equation}
In Section \ref{subsec:IEWsDSS} we construct EWs which detect entangled PPTDS states, therefore they correspond to indecomposable witnesses.

\section{Separability in diagonal symmetric states acting on $\mathbbm{C}^d \otimes \mathbbm{C}^d$.}
\label{sec:Separability}
In this section, we characterize the bipartite diagonal symmetric two-qudit states in terms of the separability and the PPT properties. We establish an equivalence between: (i) separability and the PPT property in DS states and (ii) quadratic conic optimization problems and their relaxations, respectively.

We first introduce the Dicke basis in its full generality and then we move to the two particular cases of interest to this paper: the case of $N$-qubits and the case of $2$-qudits.
One can think of the space spanned by the Dicke states as the linear subspace of $({\mathbbm C}^d)^{\otimes N}$ containing all permutationally invariant states.
\begin{defn}\label{def:DickeStatesGeneral}
Consider a multipartite Hilbert space $({\mathbbm C}^d)^{\otimes N}$ of $N$ qudits. The Dicke basis in that space consists of all vectors which are equal superpositions of $k_0$ qudits in the state $\ket{0}$, $k_1$ qudits in the state $\ket{1}$, etc., where the multiindex variable $\mathbf{k} = (k_0, \ldots, k_{d-1})$ forms a partition of $N$; \textit{i.e.}, $k_i \geq 0$ and $\sum_{i=0}^{d-1}k_i = N$. They can be written as
 \begin{equation}
  \ket{D_{\mathbf{k}}} \propto \sum_{\sigma \in {\mathfrak{S}_N}} \sigma(\ket{0}^{\otimes k_0} \ket{1}^{\otimes k_1} \cdots \ket{d-1}^{\otimes k_{d-1}}),
 \end{equation}
where $\sigma$ runs over all permutations of $N$ elements.
\end{defn}
The Dicke state has ${N \choose \mathbf{k}}$ different elements, where the quantity follows from the multinomial combinatorial quantity
\begin{equation}
 {N \choose \mathbf{k}} = \frac{N!}{k_0! k_1! \cdots k_{d-1}!}.
\end{equation}
Finally, recall that there are as many Dicke states as partitions of $N$ into $d$ (possibly empty) subsets; therefore, the dimension of the subspace of $({\mathbbm C}^d)^{\otimes N}$ is given by
\begin{equation}
 \dim [\{\ket{D_{\mathbf k}}: \mathbf{k} \vdash N\}] = {N + d-1 \choose d-1},
\end{equation}
where $\vdash$ denotes \textit{partition of}.

In this paper, we are particularly interested in the case of $N$-qubits and $2$-qudits:
\begin{itemize}
 \item $d=2$. For $N$-qubit states we shall use the notation $\ket{D_{\mathbf k}} \equiv \ket{D_{k}^N}$, where $k=k_1$ denotes the number of qubits in the excited ($\ket{1}$) state.
 Mixtures of Dicke states correspond to
\begin{eqnarray}
\rho=\sum\limits^N_{k=0}p_k\ket{D^N_k}\bra{D^N_k}.
\end{eqnarray}
 \item $N=2$. For bipartite $d$-level systems, we are going to denote the Dicke states by $\ket{D_{\mathbf k}} \equiv \ket{D_{ij}}$, where $i$ and $j$ are the indices (possibly repeated) of the non-zero $k_i$ and $k_j$. Since the terminology \textit{Dicke states} is often reserved for the multipartite case, we call $\ket{D_{ij}}$ simply \textit{symmetric states}.
\end{itemize}

In the bipartite case (Sections \ref{sec:Separability} and \ref{sec:sufficient}), we focus on diagonal symmetric states, given by Def. \ref{def:DSS}:

\begin{defn}\label{def:DSS}
Let $\rho$ be a state acting on a bipartite Hilbert space ${\cal H}_A \otimes {\cal H}_B = \mathbbm{C}^d \otimes \mathbbm{C}^d$. The state $\rho$ is said to be \emph{diagonal symmetric} (DS) if, and only if, $\rho$ can be written in the form
 \begin{equation}
  \rho = \sum_{0 \leq i \leq j < d} p_{ij}\ket{D_{ij}}\bra{D_{ij}},
  \label{eq:DSS}
 \end{equation}
where $p_{ij} \geq 0$, $\sum_{ij} p_{ij}=1$, $\ket{D_{ii}}: = \ket{ii}$ and $\ket{D_{ij}} = (\ket{ij} + \ket{ji})/\sqrt{2}$.
\end{defn}

In the computational basis, a DS $\rho$ is a $d^2 \times d^2$ matrix that is highly sparse. Therefore, it will be useful to associate a $d\times d$ matrix to $\rho$ that captures all its relevant information. We define the $M$-matrix of $\rho$ to be
\begin{defn}\label{def:M}
To every DS $\rho$ acting on $\mathbbm{C}^d \otimes \mathbbm{C}^d$, there is an associated $d \times d$ matrix $M(\rho)$, with non-negative entries
\begin{equation}
 \label{eq:def:M}
  M(\rho) := \left(
  \begin{array}{cccc}
   p_{00} & p_{01}/2 & \cdots & p_{0,d-1}/2\\
   p_{01}/2 & p_{11}& \cdots & p_{1,d-1}/2\\
   \vdots & \vdots & \ddots & \vdots\\
   p_{0,d-1}/2& p_{1,d-1}/2&\cdots&p_{d-1,d-1}
  \end{array}
  \right),
 \end{equation}
 which arises from the partially transposed matrix $\rho^{T_B}$.
\end{defn}
 
Notice that, while a DS state $\rho$ is always diagonal in the Dicke basis, its partial transposition (which is defined with respect to the computational basis) scrambles its elements. Then $\rho^{T_B}$ is block-diagonal in the Dicke basis and its blocks are $1 \times 1$ elements corresponding to $p_{ij}$ with $i < j$, and $M(\rho)$.
One can see the effect of the partial transposition operation on a DS state by inspecting the action of $T_B$ onto the elements $\ket{D_{ij}}\bra{D_{ij}}$ that compose Eq. (\ref{eq:DSS}):
\begin{itemize}
 \item If $i = j$, then $(\ket{D_{ii}}\bra{D_{ii}})^{T_B} = \ket{D_{ii}}\bra{D_{ii}}$, because $\ket{D_{ii}}= \ket{ii}$.
 \item If $i \neq j$, the action of the partial transposition is best seen by expanding $\ket{D_{ij}}$ onto the computational basis: $(\ket{D_{ij}}\bra{D_{ij}}) = \frac{1}{2}(\ket{ij}\bra{ij} + \ket{ij}\bra{ji} +\ket{ji}\bra{ij} +\ket{ji}\bra{ji})$. Therefore, two of the terms are left invariant
 and the remaining two are to be mapped as $(\ket{ij}\bra{ji}+\ket{ji}\bra{ij})^{T_B} = \ket{ii}\bra{jj}+\ket{jj}\bra{ii}$.
\end{itemize}
Thus, $M(\rho)$ is the submatrix corresponding to the elements indexed by $\ket{ii}\bra{jj}$ for $0 \leq i,j < d$ of $\rho^{T_B}$. Because there is no mixing between other rows or columns, we have that $\rho^{T_B}$ decomposes as the direct sum
\begin{equation}
 \rho^{T_B} = M(\rho) \bigoplus_{0 \leq i \neq j < d} \left(\frac{p_{ij}}{2}\right).
\end{equation}
Since $p_{ij} = p_{ji}$, we find that the $1 \times 1$ blocks appear all with multiplicity $2$.

Therefore, each $M(\rho)$ with non-negative entries summing $1$ is in one-to-one correspondence to a DS state $\rho$. In this section we characterize the separability properties of $\rho$ in terms of equivalent properties of $M(\rho)$, which are naturally related to quadratic conic optimization.

In quadratic conic optimization, one is interested in characterizing the so-called completely positive (CP) matrices, which are defined as follows
\begin{defn}
\label{def:CP}
Let $A$ be a $d\times d$ matrix. $A$ is \emph{completely positive} (CP) if, and only if, it admits a decomposition $A=B \cdot B^T$, where $B$ is a $d\times k$ matrix, for some $k\geq 1$, such that $B_{ij}\geq 0$.
\end{defn}
Matrices which are CP form a proper\footnote{Closed, convex, pointed and full-dimensional.} cone, which is denoted by ${\cal CP}_d$. Note that the sum of two CP matrices is a CP matrix and the multiplication of a CP matrix by a non-negative scalar is a CP matrix.

Given a non-convex optimization problem over the simplex, which is NP-hard in general, CP matrices translate the complexity of the problem by reformulating it as a linear problem in matrix variables over ${\cal CP}_d$. Therefore, they allow to shift all the difficulty of the original problem into the cone constraint. Precisely, every non-convex quadratic optimization problem over the simplex (LHS of Eq. (\ref{eq:CP-equivalence})) has an equivalent CP formulation (RHS of Eq. (\ref{eq:CP-equivalence})):
\begin{equation}
\max_{x_i \geq 0,\ \langle u | x \rangle = 1} \bra{x}Q \ket{x} = \max_{X \in {\cal CP}_d,\ \langle u |X|u \rangle = 1} \mathrm{Tr}(XQ),
\label{eq:CP-equivalence}
\end{equation}
where $\ket{u}$ is the unnormalized vector of ones and $Q$ is, without loss of generality\footnote{$Q$ can be assumed to be symmetric because $\bra{x}Q\ket{x} = (\bra{x}Q\ket{x})^T = \bra{x}Q^T\ket{x}$. It can be assumed to be positive semi-definite because adding $\mathbbm{1}$ to $(Q+Q^T)/2$ does not change the optimal $\ket{x}$; it only adds a bias to the maximum.}, symmetric and positive semi-definite.
Therefore, deciding membership in ${\cal CP}_d$ is NP-hard \cite{berman2015open}.

One can obtain, however, an upper bound on the optimization in Eq. (\ref{eq:CP-equivalence}) by observing that every CP matrix $A$ is positive semi-definite, because it allows for a factorization $A = B \cdot B^T$. Moreover, it is also entry-wise non-negative: $A_{ij} \geq 0$. This motivates Definition \ref{def:DNN}:
\begin{defn}
\label{def:DNN}
Let $A$ be a $d\times d$ matrix. $A$ is \emph{doubly non-negative} (DNN) if, and only if, $A\succeq 0$ and $A_{ij} \geq 0$.
\end{defn}

We are now ready to introduce the equivalences between the separability problem in DS states and quadratic conic optimization. After producing our results, we learned that these equivalences were independently observed by Nengkun Yu in \cite{KoreanGuy}. We nevertheless prove them in a different way.
\begin{theorem}
\label{thm:Winterfell}
 Let $\rho$ be a DS state acting on $\mathbbm{C}^d \otimes \mathbbm{C}^d$.
 \begin{equation}
  \rho \mbox{ is separable} \iff M(\rho) \mbox{ is CP}.
 \end{equation}
\end{theorem}

We prove Theorem \ref{thm:Winterfell} in Appendix \ref{app:ProofOfWinterfell}.

By virtue of Theorem \ref{thm:Winterfell}, we recover the result of \cite{KoreanGuy}: Because it is NP-Hard to decide whether a matrix admits a CP decomposition \cite{berman2015open}, the separability problem in $\mathbbm{C}^d \otimes \mathbbm{C}^d$ DS states is NP-Hard. 

We remark that the NP-hardness result that we obtain holds under polynomial-time Turing reductions\footnote{Intuitively speaking, a Turing reduction describes how to solve problem $A$ by running an algorithm for a second problem $B$, possibly multiple times.}, as opposed to poly-time many-one\footnote{A many-one reduction is a special case of a Turing reduction, with the particularity that the algorithm for problem $B$ can be called only one time, and its output is immediately returned as the output of problem $A$.} reductions \cite{Gharibian10}. For instance, this is the case for Gurvits' initial reduction of the weak membership problem\footnote{Weak in the sense that it allows for error in points at a given Euclidean distance from the border of the set.} in the set of separable states from the NP-complete problem PARTITION\footnote{The PARTITION problem is a decision problem corresponding to whether a given set of integer numbers can be partitioned into two sets of equal sum. This problem is efficiently solvable with a dynamic programming procedure \cite{Gareybook}, but becomes NP-hard when the magnitudes of the input integers become exponentially large with the input size.} \cite{Gurvits03}. In the case we present here, the reduction holds because the NP-hardness of deciding membership in the ${\cal CP}_d$ set follows via a Turing reduction, which is the result we use as our starting point. The part of the reduction that we provide here, however, is many-one.
\newpage

We here briefly discuss the steps that would be required to make this result completely rigurous from a computer science point of view. One would need to embed the NP-hardness into the formalism of the weak membership problem \cite{Gurvits03, grotschel2012geometric}. This requires, for instance, showing that the convex set of separable DS states has some desirable properties such as being well-bounded or $p$-centered. We refer the reader to \cite{Gharibian10} for the technical aspects of these definitions. On the other hand, the completely positive cone is known to be well-bounded and p-centered: in \cite{dickinson2013scaling} it was proved that the weak membership problem in the completely positive cone is NP-hard. By using the one-to-one correspondence between DS states and CP matrices given by $M(\rho)$ in Def. \ref{def:M}, then the result is mapped onto the DS set\footnote{The technical requirement of full dimensionality \cite{Gharibian10, dickinson2013scaling} depends on the set in which one embeds the problem. Recall that we are interested in solving the separability problem within DS states. The set of DS separable states is of course not full-dimensional when embedded in the whole two-qudit Hilbert space. However, it is full-dimensional when viewed in the DS subspace (cf. Figure \ref{fig:DS}).}.

Geometrically, the set of separable DS states is convex. Hence, it is fully characterized by its extremal elements (those that cannot be written as a non-trivial convex combination of other separable DS states). Identifying such elements is of great importance towards the characterization of the separability properties of DS states. For instance, in the set of all separable density matrices, the extremal ones are the rank-$1$ projectors onto product vectors. However, this property may be lost when restricting our search in a subspace: observe that the set of separable DS states states is obtained as the intersection of the subspace of DS states with the convex set of separable density matrices. Therefore, the set of extremal separable DS states states may contain states that are separable, but not extremal in the set of separable density matrices (see Fig. \ref{fig:DS}). Theorem \ref{thm:Winterfell} allows us to fully characterize extremality in the set of separable DS states in terms of extremal CP matrices, thus obtaining the following corollary:

\begin{cor}
\label{cor:extr}
The extremal separable DS states $\rho$ fulfill
 \begin{equation}
 p_{ij} = 2\sqrt{p_{ii}p_{jj}}, \qquad i < j.
 \label{eq:Bran}
 \end{equation}
\end{cor}

\paragraph{Proof.} -- 
Since the extremal rays of the ${\cal CP}_d$ cone are the rank-$1$ matrices $\vec{b}\ \vec{b}^T$ where $b_i \geq 0$ \cite{berman2003completely}, by normalizing and comparing to Eq. (\ref{eq:def:M}) we obtain Eq. (\ref{eq:Bran}).

\begin{figure}[h]
\centering
\includegraphics[scale=1]{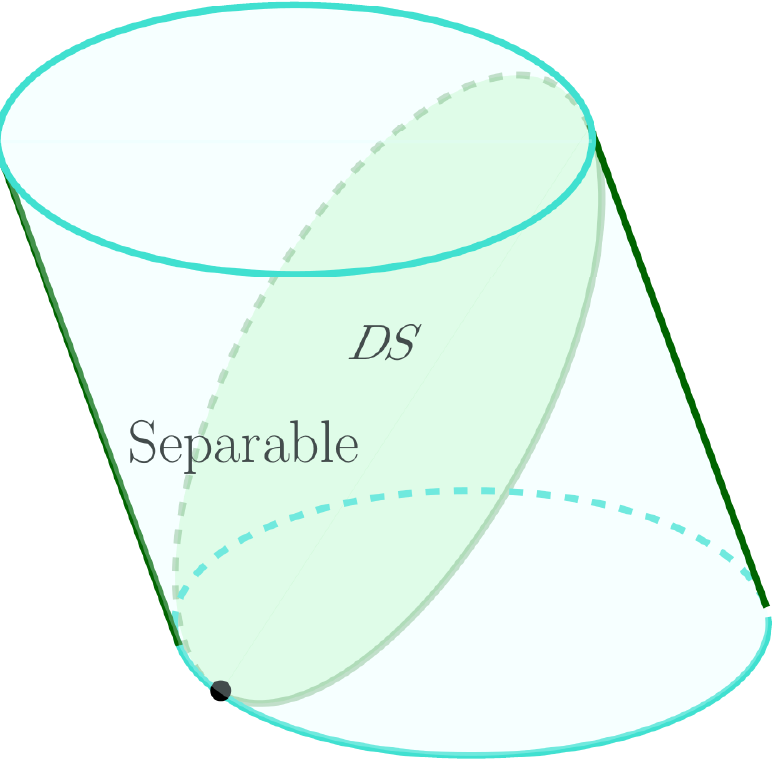}
\caption{Cartoon picture of the set of separable states SEP (cylinder) and its intersection with the subspace of diagonal symmetric states DS (ellipse). The intersection of the subspace of DS states with the set of separable states gives rise to the set of separable DS states, which is represented by the green ellipse, including its interior.
Only the states of the form $\ket{ii}\bra{ii}$ are extremal in both sets (represented by the black dot in the figure). However, states that were in the boundary of SEP, could now be extremal when viewed in DS (represented by the border of the green ellipse in the figure).}
\label{fig:DS}
\end{figure}

\begin{theorem}
\label{thm:KingsLanding}
 Let $\rho$ be a DS state acting on $\mathbbm{C}^d \otimes \mathbbm{C}^d$.
 \begin{equation}
  \rho \mbox{ is PPT} \iff M(\rho) \mbox{ is DNN}.
 \end{equation}
\end{theorem}

\paragraph{Proof.} -- 
Let us assume that $\rho$ is PPT. Note that the partial transposition of $\rho$ can be written as
\begin{equation}
\label{eq:LittleFinger}
\rho^{\Gamma} = \left(\bigoplus_{0\leq a < b < d} \left(\frac{p_{ab}}{2}\right)\oplus\left(\frac{p_{ab}}{2}\right)\right)\oplus M(\rho).
\end{equation}
Since $\rho$ is PPT, Eq. (\ref{eq:LittleFinger}) implies that $M(\rho)\succeq 0$. Since $\rho$ is a valid quantum state, then $p_{ab} \geq 0$ for all $0\leq a \leq b < d$. Hence, all the entries of $M(\rho)$ are also non-negative. Thus, $M(\rho)$ is DNN.

Conversely, if $M(\rho)$ is DNN then we have that all its entries are non-negative; \textit{i.e.,} $p_{ab} \geq 0$ for $0 \leq a \leq b < d$. These conditions guarantee that $\rho \succeq 0$. Additionally, as $M(\rho) \succeq 0$, these conditions imply that $\rho^\Gamma \succeq 0$. Hence, $\rho$ is PPT.

\qed

\begin{figure}[h]
\centering
\includegraphics[scale=0.5]{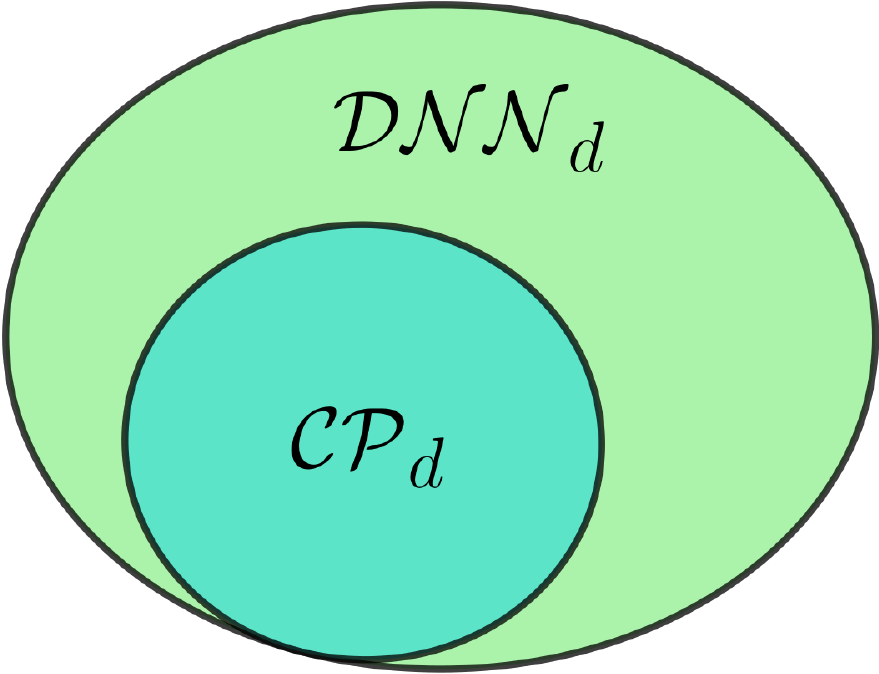}
\caption{For a two qudit PPTDS state $\rho$, if its corresponding $M(\rho)$ is in $\mathcal{CP}_d$ then $\rho$ is separable, if $M(\rho)$ is in $\mathcal{DNN}_d$ then $\rho$ is $PPT$ and if $M(\rho)$ is in $\mathcal{DNN}_d$ but not $\mathcal{CP}_d$ then $\rho$ is $PPT$ but entangled.}
\label{fig:sets}
\end{figure}

Recall (cf. Definitions \ref{def:CP} and \ref{def:DNN}, also Fig. \ref{fig:sets}) that ${\cal CP}_d \subseteq {\cal DNN}_d$. However, the inclusion is strict for $d\geq 5$:
It is known that ${\cal CP}_d = {\cal DNN}_d$ for $d\leq 4$ and ${\cal CP}_d \subsetneq {\cal DNN}_d$ for $d \geq 5$ \cite{berman2015open}. This yields a full characterization of the bipartite separable DS states in terms of the PPT criterion:

\begin{theorem}
\label{thm:TheEyre}
 Let $\rho$ be a DS state acting on $\mathbbm{C}^d \otimes \mathbbm{C}^d$, with $d\leq 4$.
 \begin{equation}
  \rho \mbox{ is separable} \iff \rho \mbox{ is PPT}.
 \end{equation}
\end{theorem}

\paragraph{Proof.} -- The result follows from the identity ${\cal CP}_d = {\cal DNN}_d$, which holds for $d \leq 4$ \cite{berman2015open}. In Example \ref{ex:PPT3} we provide a constructive proof for $d=3$.

\qed

Finally, we end this section by giving a sufficient separability criteria for any $d$ in terms of the ranks of $M(\rho)$.

\begin{theorem}\label{thm:PPTDSsep}
Let $\rho$ be a PPTDS state with $M(\rho)$ of rank at most $2$. Then, $\rho$ is separable.
\end{theorem}

\paragraph{Proof.} --
Since $\rho$ is PPT, $M(\rho)\succeq 0$. Therefore, it admits a factorization $M(\rho) = V V^T$, where $V$ is a $d \times 2$ or a $d \times 1$ matrix. Geometrically, every row of $V$ can be seen as a vector in $\mathbbm{R}^2$ (or a scalar if the rank of $M(\rho)$ is one). Therefore, $M(\rho)$ can be seen as the Gram matrix of those vectors; each element being their scalar product. Since $M(\rho)$ is doubly non-negative, it implies that all these scalar products must be positive; therefore, the angle between each pair of vectors is smaller or equal than $\pi/2$. Thus, the geometrical interpretation is that $M(\rho)$ is CP if, and only if, they can be isometrically embedded into the positive orthant of $\mathbbm{R}^k$ for some $k$. This is always possible to do for $k = 2$ (see Fig. \ref{fig:geom_interp}), which corresponds to $M(\rho)$ having rank at most $2$.

\qed

\begin{figure}[h]
\centering
\includegraphics{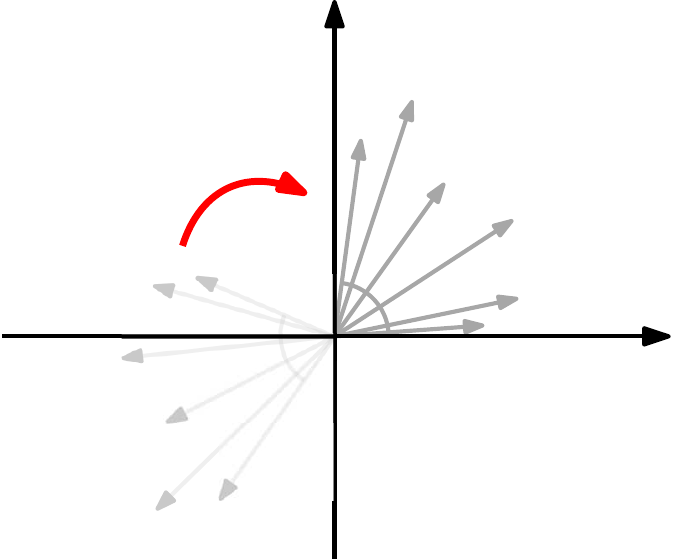}
\caption{Visual representation of the proof for Theorem \ref{thm:PPTDSsep}. When the angle between each pair of vectors is smaller or equal than $\pi/2$, meaning that $M(\rho)$ is $\mathcal{CP}_d$, all vectors can be isometrically embedded in the nonnegative orthant.}
\label{fig:geom_interp}
\end{figure}

\section{Sufficient criteria for entanglement and separability}
\label{sec:sufficient}
In this section we further characterize the bipartite DS states by providing sufficient criteria to certify entanglement by means of Entanglement Witnesses for DS states, and by providing sufficient separability conditions in terms of $M(\rho)$.
\subsection{Entanglement Witnesses for DS states}
\label{subsec:IEWsDSS}

We begin by introducing the concept of \emph{copositive} matrix:
\begin{defn}
A matrix $A$ is called \emph{copositive} if, and only if, $\vec{x}^T A \vec{x} \geq 0$ for all $\vec{x}$ with non-negative entries.
\end{defn}

The set of $d\times d$ copositive matrices also forms a proper cone, denoted ${\cal COP}_d$. The cones ${\cal CP}_d$ and ${\cal COP}_d$ are dual to each other with respect to the trace inner product. It is also easy to see that ${\cal PSD}_d + {\cal N}_d \subseteq {\cal COP}_d$, where ${\cal PSD}_d$ is the set of positive-semidefinite $d\times d$ matrices and ${\cal N}_d$ is the set of symmetric entry-wise non-negative matrix. Actually, we have ${\cal DNN}_d = {\cal PSD}_d \cap {\cal N}_d$ and the observation follows from the inclusion ${\cal CP}_d \subseteq {\cal DNN}_d$.

Therefore, one can view copositive matrices as EWs for DS states. Furthermore, one could think of ${\cal PSD}_d + {\cal N}_d$ as the set of DEWs for DS states, in the sense that they do not detect entangled PPTDS states.\\
In Examples \ref{subsec:Exemples5x5} and \ref{ex:5x5second} we provide some $M(\rho) \in {\cal DNN}_d \setminus {\cal CP}_d$ for $d\geq 5$, therefore corresponding to entangled PPTDS states. The paradigmatic example of a copositive matrix detecting matrices DNN, but not CP, (\textit{i.e.}, PPT, but entangled DS states) is the \emph{Horn} matrix \cite{HornMatrix}, which is defined as
\begin{equation}
H := \left(
\begin{array}{rrrrr}
1&-1&1&1&-1\\
-1&1&-1&1&1\\
1&-1&1&-1&1\\
1&1&-1&1&-1\\
-1&1&1&-1&1
\end{array}
\right).
\end{equation}
It is proven that $H \in {\cal COP}_5 \setminus ({\cal PSD}_5 + {\cal N}_5)$ in \cite{HornMatrix}. As $\mathrm{Tr}(H M(\rho))=-1< 0$, $H$ corresponds to an (indecomposable) entanglement witness for the state corresponding to $M(\rho)$.

Although the boundary of the set of copositive matrices remains uncharacterized for arbitrary dimensions, ${\cal COP}_5$ was fully characterized in \cite{dickinson2013scaling}:
\begin{equation}
{\cal COP}_5 = \{DAD : D \mbox{ is positive diagonal}, A \mbox{ s. t. } p(A,\vec{x}) \mbox{ is a sum of squares}\},
\end{equation}
where $$p(A,\vec{x}):=\left(\sum_{i,j} A_{ij} x_i^2x_j^2\right)\left(\sum_{k} x_k^2\right).$$
Furthermore, the extremal rays of ${\cal COP}_d$ have been fully characterized for $d\leq 5$, divided into classes, but this also remains an open problem for higher $d$ \cite{berman2015open}.

In Appendix \ref{app:Exposedness} we discuss exposedness properties of the sets of completely positive and co-positive matrices and their relation to the separability problem and its geometry.

\paragraph{Examples of entangled PPTDS states for $d=5$.} --
\label{subsec:Exemples5x5}
Let us provide an example of a bipartite PPTDS entangled state for $d=5$.
\begin{equation}
\tilde{M}(\rho)=\left(
\begin{array}{ccccc}
1&1&0&0&1\\
1&2&1&0&0\\
0&1&2&1&0\\
0&0&1&1&1\\
1&0&0&1&3
\end{array}
\right).
\end{equation}
It can be easily seen using the Range criterion, as in Section \ref{ex:d6}, that $\rho$ is entangled. By Theorem \ref{thm:Winterfell}, it is equivalent to show \cite{berman2015open} that $M(\rho) \in {\cal DNN}_5 \setminus {\cal CP}_5$.\\
Finally, it can be appreciated how the Horn matrix can be used as an EW and certify entanglement $\mathrm{Tr}(H \tilde{M}(\rho))=-1< 0$.

\subsection{Sufficient separability conditions for diagonal symmetric states}
In the spirit of best separable approximations (BSA) \cite{LewensteinSanperaSep}, in this section we provide sufficient separability conditions for bipartite PPTDS states. In the same way that the BSA allows one to express any PPTDS state as a sum of a separable part and an entangled one with maximal weight on the separable one\footnote{In \cite{QuesadaBSA}, the BSA was found analytically for $N$-qubit DS states}. In this section we introduce Best Diagonal Dominant (BDD) approximations, which give a sufficient criterion to certify that a PPTDS state is separable. The idea is that although checking membership in ${\cal CP}_d$ is NP-hard, it is actually easy to (i) characterize the extremal elements in ${\cal CP}_d$ (cf. Corollary \ref{cor:extr}) and (ii) check for membership in a subset of ${\cal DD}_d \subseteq {\cal CP}_d$ that is formed of those matrices $A \in {N}_d$ that are diagonal dominant. In \cite{kaykobad1987nonnegative} the inclusion ${\cal DD}_d \subseteq {\cal CP}_d$ was proven. Therefore, to show that ${\cal CP}_d \setminus {\cal DD}_d$ is nonempty we study when the decomposition of a potential element in ${\cal CP}_d$ as a convex combination of an extremal element of ${\cal CP}_d$ and an element of ${\cal DD}_d$ is possible (see Figure \ref{fig:1}).

Let us start by stating a lemma that gives an explicit separable decomposition of a quantum state.
\begin{lemma}
\label{lem:I}
Let $I$ be the unnormalized quantum state defined as 
\begin{equation}
I = \sum_{i=0}^{d-1} \ket{ii}\bra{ii} + \sum_{0 \leq i< j < d} 2\ket{D_{ij}}\bra{D_{ji}},
\end{equation}
where $\ket{D_{ij}} = (\ket{ij} + \ket{ji})/\sqrt{2}$. For instance, for $d=3$,
\begin{equation}
I = \left(
\begin{array}{ccc|ccc|ccc}
1&\cdot&\cdot&\cdot&\cdot&\cdot&\cdot&\cdot&\cdot\\
\cdot&1&\cdot&1&\cdot&\cdot&\cdot&\cdot&\cdot\\
\cdot&\cdot&1&\cdot&\cdot&\cdot&1&\cdot&\cdot\\
\hline
\cdot&1&\cdot&1&\cdot&\cdot&\cdot&\cdot&\cdot\\
\cdot&\cdot&\cdot&\cdot&1&\cdot&\cdot&\cdot&\cdot\\
\cdot&\cdot&\cdot&\cdot&\cdot&1&\cdot&1&\cdot\\
\hline
\cdot&\cdot&1&\cdot&\cdot&\cdot&1&\cdot&\cdot\\
\cdot&\cdot&\cdot&\cdot&\cdot&1&\cdot&1&\cdot\\
\cdot&\cdot&\cdot&\cdot&\cdot&\cdot&\cdot&\cdot&1
\end{array}
\right).
\end{equation}
Then, $I$ is separable.
\end{lemma}
\paragraph{Proof.} -- 
Let $\ket{e(\vec{\varphi})} = \ket{0} + e^{\mathbbm{i} \varphi_1}\ket{1} + \cdots + e^{\mathbbm{i} \varphi_{d-1}}\ket{d-1}$. A separable decomposition of $I$ is given by
\begin{equation}
I = \int_{[0, 2\pi]^{d}} \frac{\mathrm{d}\vec{\varphi}}{(2 \pi)^d} (\ket{e(\vec{\varphi})}\bra{e(\vec{\varphi})})^{\otimes 2}.
\end{equation}

Indeed,
\begin{equation}
I = \sum_{ijlk} \ket{ij}\bra{kl} \int_{[0, 2\pi]^{d}} \frac{\mathrm{d}\vec{\varphi}}{(2 \pi)^d} e^{\mathbbm{i}(\varphi_{i} + \varphi_{j} - \varphi_{k} - \varphi_{l})} = \sum_{ijlk} \ket{ij}\bra{kl} (\delta_{i,k}\delta_{j,l} + \delta_{i,l}\delta_{j,k} - \delta_{i,j,k,l}),
\end{equation}
where $\delta$ is the Kronecker delta function.
\qed

Lemma \ref{lem:I} allows us to give a sufficient condition for a state $\rho$ to be separable. The idea is to subtract $\varepsilon I$ from $\rho$ in such a way that it remains a valid diagonal symmetric state, PPT, and close enough to the interior of the separable set such that it is easy to certify that the state is separable (see Fig. \ref{fig:1}).

\begin{figure}
\centering
\includegraphics[scale=0.75]{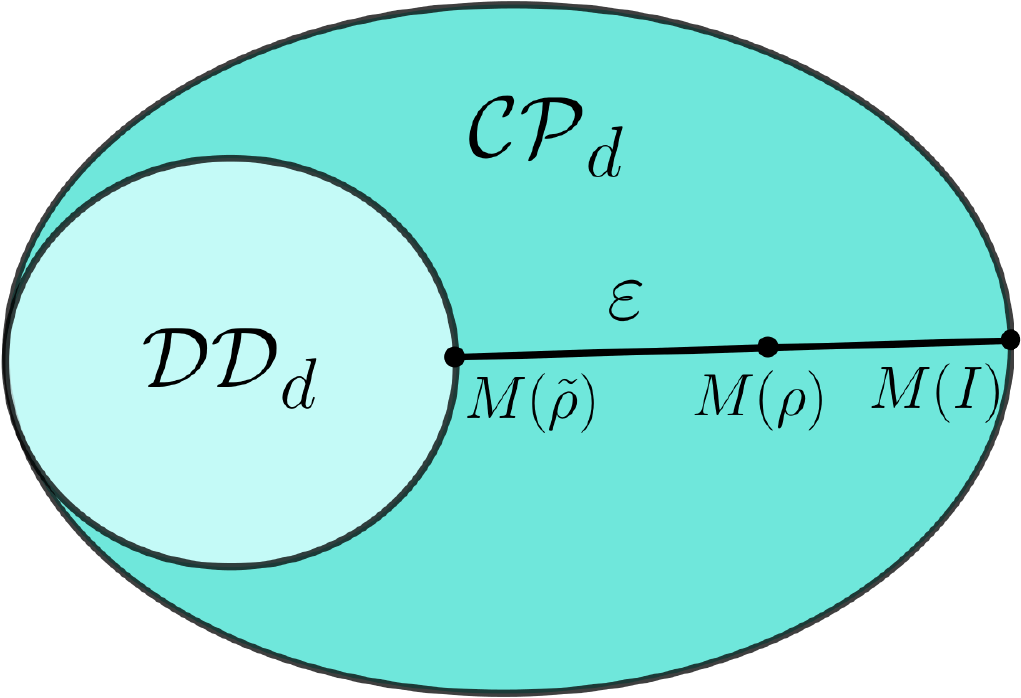}
\caption{Given a two-qudit PPTDS state $\rho$, if we can decompose $M(\rho)$ in terms of $M(I)$ (an extremal element of $\mathcal{CP}_d$) and $M(\tilde{\rho})$ (an element of $\mathcal{DD}_d$) we can certify that $M(\rho)$ is in $\mathcal{CP}_d$ and therefore certify that $\rho$ is separable.}
\label{fig:1}
\end{figure}

\begin{theorem}
\label{thm:Anna}
Let $\rho$ be a two-qudit PPTDS state with associated $M(\rho)$. If there exists $\varepsilon \geq 0$ such that
\begin{enumerate}
\item $\varepsilon \leq \rho_{ij}$ for all $i$, $j$ such that $0 \leq i,j < d$.
\item $\varepsilon d \leq (\bra{u} \frac{1}{M(\rho)}\ket{u})^{-1}$ and $\ket{u} \in {\cal R}(M(\rho))$, ${\cal R}(M(\rho))$ is the range of $M(\rho)$ and $\frac{1}{M(\rho)}$ is the pseudo-inverse of $M(\rho)$. Here $\ket{u}$ is a normalized vector of ones.
\item for all $i$ such that $0 \leq i < d$, $\rho_{ii} + \varepsilon (d-2) \geq \sum_{j \neq i}\rho_{ji}$.
\end{enumerate}
Then, $\rho$ is separable.
\end{theorem}

See the proof in Appendix \ref{proof:Annathm}.

A few comments are in order: The first condition on Theorem \ref{thm:Anna} ensures that $I$ can be subtracted from $\rho$ and $\tilde{\rho}$ will remain in the DS subspace. The second condition requires that $I$ can be subtracted from $\rho$ while maintaining the PPT property of $\tilde \rho$. If $\ket{u} \notin {\cal R}(M(\rho))$, then $\tilde \rho$ would not be PPT for any $\varepsilon \neq 0$. Therefore, the second condition gives the maximal value of $\varepsilon$ that can be subtracted such that $\tilde \rho$ remains PPT. Finally, the third condition relies on guaranteeing that $\tilde \rho$ is separable, which is ensured by $M(\tilde \rho)$ to be diagonal dominant. This means that one might need to subtract a minimal amount of $I$ to accomplish such a property (unless $\rho$ is already diagonal dominant).\\
In Example \ref{ex:Albert1} we show that ${\cal CP}_d \setminus {\cal DD}_d \neq \emptyset$ using the approach of Theorem \ref{thm:Anna}.\\

The above result can be now normalized and generalized to any other extremal element $I$ in ${\cal CP}_d$:

\begin{lemma}
  Let $\mathbf{x} \in \mathbbm{R}^d$ with $x_i \geq 0$ and $||\mathbf{x}|| > 0$. Let $I_{\mathbf x}$ be the quantum state defined as $$I_{\mathbf x} = \frac{1}{||\mathbf{x}||_1^2}\left( \sum_{i=0}^{d-1} x_i^2 \ket{ii}\bra{ii} + \sum_{0 \leq i < j < d} 2 x_i x_j\ket{D_{ij}}\bra{D_{ij}}\right).$$ 
  Then, the quantum state $I_\mathbf{x}$ is separable.
  \label{lem:X}
\end{lemma}

See the proof in Appendix \ref{proof:lemmaIx}.

Note that the corresponding $M(I_{\mathbf x})$ is given by $\ket{u_{\mathbf x}}\bra{u_{\mathbf x}}$, where $\ket{u_{\mathbf x}} = \mathbf{x}/||\mathbf{x}||_1$. The sum of the elements of $M(I_{\mathbf x})$ is then one; \textit{i.e.}, $||M(I_{\mathbf x})||_1 = 1$.

Lemma \ref{lem:X} allows us to give a sufficient condition for a state $\rho$ is separable. This time the idea is to decompose $\rho$ as a convex combination between $I_{\mathbf x}$, which is a state that is extremal in the set of separable DS states, and a state $\tilde \rho$ which is deep enough in the interior of the set of separable states, such that we can certify its separability by other means (by showing that $M(\tilde{\rho})$ is diagonal dominant and doubly non-negative; therefore completely positive \cite{kaykobad1987nonnegative}).

\begin{theorem}
\label{thm:JordisGen}
 Let $\rho$ be a two-qudit PPTDS state with associated $M(\rho)$. Let $\mathbf{x} \in \mathbbm{R}^d$ with $x_i > 0$. If there exists a $\lambda \in [0, 1)$ such that
 \begin{enumerate}
  \item $\lambda \leq (M(\rho))_{ij} ||\mathbf{x}||_1^2/x_i x_j$ for all $i$ and $j$,
  \item $\lambda \leq 1/ \bra{u_{\mathbf{x}}} \frac{1}{M(\rho)} \ket{u_{\mathbf{x}}}$ and $\ket{u_{\mathbf x}} \in {\cal R}(M(\rho))$, where ${\cal R}(M(\rho))$ is the range of $M(\rho)$ and $\frac{1}{M(\rho)}$ denotes the pseudo-inverse of $M(\rho)$,
  \item $\lambda x_i (||\mathbf{x}||_1 - 2 x_i) \geq ||\mathbf{x}||_1^2 \left[ \sum_{j \neq i} (M(\rho))_{ij} - (M(\rho))_{ii}\right]$ for all $i$,
 \end{enumerate}
then $\rho$ is separable. Equivalently, then $M(\rho)$ is completely positive.
\end{theorem}

See the proof in Appendix \ref{proof:JordisGen} (\textit{i.e.,} write $\rho=(1-\lambda)\tilde{\rho}+\lambda I_\textbf{x}$ and ensure that the associated $M(\tilde{\rho})$ is completely positive).\\
Notice that Theorem \ref{thm:JordisGen} provides an advantage over Theorem \ref{thm:Anna} since the parameters of $I_\textbf{x}$ are not fixed which allows to consider a bigger family of decompositions $M(\rho)\in\mathcal{CP}_d$ parametrized by $\textbf{x}$. In Example \ref{ex:Albert2} we attempt to apply both Theorems in order to guarantee separability and illustrate such advantage.

\paragraph{Example.} --
\label{ex:Albert2}
In this example we provide a PPTDS state with associated $M(\rho)\in \mathcal{CP}_d\setminus \mathcal{DD}_d$ and we show how to apply Theorem \ref{thm:JordisGen} to guarantee separability. Furthermore, we also apply Theorem \ref{thm:Anna} to illustrate the advantage of Theorem \ref{thm:JordisGen}.\\
Take the following DS state $\rho\in \mathbbm{C}^3$ with associated
\begin{equation}\label{eq:example2}
M(\rho)=\left(
\begin{array}{ccc}
\alpha&\beta&\gamma\\
\beta&\delta&\beta\\
\gamma&\beta&\epsilon\\
\end{array}
\right)=\frac{1}{100}\left(
\begin{array}{ccc}
19&8&11.5\\
8&6.4&8\\
11.5&8&19.6\\
\end{array}
\right),
\end{equation}
where it can be checked that the state $\rho$ is normalized.
A priori we do not know if this state is separable, the goal is to apply Theorems \ref{thm:Anna} and \ref{thm:JordisGen} in order to see if separability can be guaranteed. For both Theorems the more restrictive between conditions 1 and 2 provides an upper bound for the corresponding decomposition and condition 3 a lower bound but, as mentioned, Theorem \ref{thm:JordisGen} offers more flexibility since such bound can be varied by fitting $I_\textbf{x}$. This example illustrates this fact since we will see that Theorem \ref{thm:JordisGen} guarantees separability but Theorem \ref{thm:Anna} does not.  \\
Lets start with Theorem \ref{thm:JordisGen}. For the given case \eqref{eq:example2} we want to find if a convex decomposition $M(\rho)=(1-\lambda)M(\tilde{\rho})+\lambda M(I_\textbf{x})$ exists while fulfilling the conditions of the theorem. For instance, for illustrative purposes we fix $\lambda=0.8$ and by numerical means we obtain that a possible convex combination would be with an $M(I_\textbf{x})=\ket{u_\textbf{x}}\bra{u_\textbf{x}}$ given by $\ket{u_\textbf{x}} = 1/100(37.46\ket{0} + 25.16\ket{1} + 
37.38\ket{2})$. Lets proceed to show that the given $M(\rho)$ and $M(I_\textbf{x})$ meet the conditions of Theorem \ref{thm:JordisGen}. Condition 1 provides the more restrictive upper bound given by
\begin{equation}
\lambda \leq \frac{\gamma ||\textbf{x}||_1^2}{x_1x_3}=0.8213,
\end{equation}
while the lower bound will be given by the following case of Condition 3
\begin{equation}
\lambda \geq \frac{||\textbf{x}||_1^2[2\beta-\delta]}{x_2(||\textbf{x}||_1-\delta)}=0.7681.
\end{equation}
Therefore, there exists a range of values $\lambda \in [0.7681,0.8213]$ that satisfy the conditions of Theorem \ref{thm:JordisGen} and certifies that the state $\rho$ is separable. Notice that, for illustrative purposes, once we found an $M(I_\textbf{x})$ fulfilling the conditions we fixed it to find a range of values for $\lambda$ but we could have allowed for more freedom and find a bigger range of possible decompositions.\\

Now lets see what happens with Theorem \ref{thm:Anna}. In this case the most restrictive upper bound is given by Condition 2 
\begin{equation}
\varepsilon \leq (\bra{u} \frac{1}{M(\rho)}\ket{u})^{-1}=0.06,
\end{equation}
while the most restrictive lower bound will be given by the following case of Condition 3 
\begin{equation}
\varepsilon \geq 2\beta-\delta=0.096.
\end{equation} 
Thus, for this case there does not exist an $\varepsilon$ satisfying the conditions for Theorem \ref{thm:Anna}, while there exists a range of $\lambda$ satisfying conditions for Theorem \ref{thm:JordisGen} and therefore illustrating its advantage by being able to certify separability.

\section{A class of PPT-entangled quasi-DS states}
\label{sec:Class}
In this Section, we introduce a uni-parametric class of $N$-qubit PPTESS, for an odd number of qubits. As it has been shown in \cite{KoreanGuy, QuesadaRanaSanpera}, $N$-qubit PPTDS states are fully separable. The class we introduce can be seen as a $N$-qubit PPTDS state with slight GHZ coherences. Surprisingly, in the family of states we provide, an arbitrarily small weight on the non-diagonal elements (in the Dicke basis) allows the state to be genuinely multipartite entangled while maintaining the PPT property.

The procedure we have chosen to derive this class of states is based on the iterative algorithm for finding extremal PPT symmetric states \cite{TuraPRA12, AugusiakPRA12} (see also \cite{Leinaas}), which we briefly recall here in the interest of completeness. One starts with an initial symmetric state $\rho_0$ that is fully separable; for instance, the symmetric completely mixed state. Then, one picks a random direction $\sigma_0$ in the set of quantum states and subtracts it from the initial state while preserving the PPT property, therefore obtaining $\rho_0 - x_0 \sigma_0$, $x_0 > 0$. One necessarily finds a critical $x_0^*$ such that one arrives at the boundary of the PPT set, where the rank of $\rho_0 - x_0^* \sigma_0$ or one of its partial transpositions must have decreased. Hence, at least one new vector appears in the kernel of the state or in the kernel of some of its partial transpositions. This state with lower ranks is set as the initial state for the next iteration $\rho_1 = \rho_0 - x_0^* \sigma_0$. The new direction $\sigma_1$ is chosen such that it preserves all the vectors present in the kernels of both the state and its partial transpositions. This process is repeated until no new improving direction can be found, yielding an extremal state $\rho_k$ in the PPT set. As the PPT set contains all separable states, we note that if the rank of such extremal PPT state is greater than one, then it cannot be extremal in the set of separable states (because these are pure product vectors, which have rank one), therefore it must be entangled. The study carried out in \cite{TuraPRA12, AugusiakPRA12} looked for typical extremal PPT states by exploring random directions every time. However, by carefully picking these directions, one can look for classes of states of different forms, such as the ones presented in Theorem \ref{thm:class}.

In Example \ref{ex:4qubits} we present a $4$-qubit PPT-entangled symmetric state whose density matrix is sparse with real entries when represented in the computational basis and has a closed analytical form.

The class of states we are going to present is furthermore symmetric with respect to the $\ket{0} \leftrightarrow \ket{1}$ exchange and, to simplify our proof and take advantage of this symmetry, we shall consider only an odd number of qubits $N=2K +1$, with $K>1$.
\begin{theorem}
\label{thm:class}
Let $N=2K+1$ with $K\in \mathbbm{N}$ and $K>1$.  Let $Z\in (0, \infty)$ and $\sigma = \pm 1$. We define the sequence $\{f_k(Z)\}_{k \in \mathbbm{Z}}$ through the recurrence relation $f_{k+2}(Z) = (2+Z) f_{k+1}(Z) - f_k(Z)$ and the initial conditions $f_0(Z)=1$ and $f_1(Z) = 1+Z$. We also define $\lambda_k(Z):=f_{K-k}(Z)$.
The diagonal part of the state is defined as
\begin{equation}
d(Z):=\sum_{k=0}^N {N \choose k} \lambda_k(Z) \ket{D_N^k}\bra{D_N^k},
\end{equation}
and the off-diagonal part as
\begin{equation}
o(\sigma):=\sigma(\ket{D_N^0}\bra{D_N^N} + \ket{D_N^N}\bra{D_N^0}).
\end{equation}
Then, the $N$-qubit symmetric state
 \begin{equation}
  \rho(Z):=\frac{d(Z)+o(\sigma)}{2(4+Z)^K},
 \end{equation}is PPT with respect to every bipartition, has ranks $(N+1, 2N, 2N, \ldots, 2N, 2N-1)$ and is extreme in the PPT set (therefore, it is entangled).
 \end{theorem}

We split the proof into several lemmas, for better readability and intuition on the above definitions.

\paragraph{Proof.} --

Let us start with some general considerations on the structure of $\rho(Z)$. In order to efficiently apply the partial transposition operation with respect to $m$ subsystems, we need to express $\rho(Z)$ acting on $\mathbbm{C}^{m+1}\otimes \mathbbm{C}^{N-m+1}$.
\begin{lemma}
\label{lem:blocks}
Let $\lambda_k' = {N \choose k}\lambda_k$. Let $\rho$ be an $N$-qubit PPTDS state with diagonal elements $\{\lambda_k'\}_{k=0}^N$ and GHZ coherences $o(\sigma)$. Its partial transposition with respect to $m$ subsystems $\rho^{\Gamma_m}$, acting on $\mathbbm{C}^{m+1}\otimes \mathbbm{C}^{N-m+1}$, block-diagonalizes as
 \begin{equation}
  \rho^{\Gamma_m}=\left(
  \begin{array}{cc}
    \lambda_m & \sigma\\
    \sigma & \lambda_{N-m}
  \end{array}
  \right)\bigoplus \left(
    \bigoplus_{n=-m+1}^{N-m-1} A_{n}^{(m)}
    \right),
    \label{eq:BlocksPT}
 \end{equation}
 where $A_n^{(m)} := D_{n}^{(m)} H_{n}^{(m)} D_{n}^{(m)}$. $A_n^{(m)}$ is a square matrix of size $\min \{m,N-m+n\}-\max\{0,n\}+1$ elements, which decomposes as a product of the diagonal matrix $D_n^{(m)}$, with diagonal elements
 \begin{equation}
  \left\{\sqrt{{m \choose k}{N-m \choose k-n}}\right\}_{k=\max\{0,n\}}^{\min\{m,N-m+n\}},
 \end{equation}
 and the Hankel matrix
 \begin{equation}
  H_n^{(m)} = (\lambda_{i+k-n})_{\max\{0,n\} \leq i,k \leq \min\{m,N-m+n\}}.
 \end{equation}
\end{lemma}
See the Proof in Appendix \ref{proof:blocks}.

In what follows we are going to argue the construction of our class of states. Having a state with ranks as low as possible tremendously simplifies the analysis of PPT entanglement \cite{Horodecki3x3-2x4}. It is the main idea we are going to follow in defining all the elements of our class. Therefore, we first study the condition for which the block of $\rho^{\Gamma_m}$ that contains $\sigma$ has zero determinant. This gives the condition $\lambda_m \lambda_{N-m} = \sigma^2 = 1$. In particular, if we impose this condition for $m=K$, we obtain $\lambda_K^2=1$, which means (by definition) that $f_0 = 1$.

 Now we focus on the block $n=-m+1$ with $m=K$. The determinant of the $n$-th block is
 \begin{equation}
  |A_{-K+1}^{(K)}| = K(N-K) |B_{-K+1}^{(K)}| = K(N-K) \left|
  \begin{array}{cc}
   \lambda_{K-1} & \lambda_{K}\\ \lambda_{K} & \lambda_{K+1}
  \end{array}
\right| = K(N-K) (f_1 f_0 - f_0^2) = K(N-K)(f_1-1).
 \end{equation}
By imposing the determinant of $B_{-K+1}^{(K)}$ to be $Z>0$ we obtain the condition $f_1=1+Z$ with $Z>0$. This choice is arbitrary and is the one that characterizes our class.

Finally, we move to the $3\times 3$ block determinants, which we shall make $0$. These are
\begin{equation}
|B_{-m+2}^{(m)}|=
\left|
\begin{array}{ccc}
 \lambda_{m-2} & \lambda_{m-1} & \lambda_{m}\\
 \lambda_{m-1} & \lambda_{m} & \lambda_{m+1}\\
 \lambda_{m} & \lambda_{m+1} & \lambda_{m+2}
\end{array}
\right|=0.
\end{equation}
This condition reads
\begin{equation}
 \lambda_{m+2}(\lambda_{m}\lambda_{m-2}-\lambda_{m-1}^2) = \lambda_m^3 - 2 \lambda_{m-1}\lambda_m\lambda_{m+1} + \lambda_{m-2}\lambda_{m+1}^2.
 \label{eq:recurrence}
\end{equation}
We want to determine a recursive form for the $\lambda_m$ that satisfies Eq. (\ref{eq:recurrence}). This is an equation in differences which is nonlinear. It is in general extremely hard to solve these kind of equations.
Nevertheless, despite the appearance of Eq. (\ref{eq:recurrence}), one can exploit some properties: for instance, it is immediate to check that it is homogeneous of degree $3$, its coefficients sum zero on each side of the equation and the indices of each monomial sum $3m$ for all the terms. This suggests that the equation admits a solution of the form $\lambda_{m+1} \propto \lambda_m$. Indeed, any sequence of the form $\lambda_{m+1} = c \lambda_m$ with $c\in \mathbbm{R}$ is a solution. More generally, one can show that \textit{any} sequence $\lambda_{m+2} + c_1 \lambda_{m+1} + c_0 \lambda_m = 0$ is a solution, for all $c_0, c_1 \in \mathbbm{R}$. We are going to find a solution of this form, so we have to determine $c_0$ and $c_1$.

Thanks to the symmetry $\lambda_m = \lambda_{N-m}$ we can find the coefficients $c_0$ and $c_1$: Indeed, by taking $m=K$ and $m=K-1$ we have the equations
\begin{equation}
 \left\{
 \begin{array}{ccc}
  \lambda_{K+2} + c_0 \lambda_{K+1} + c_0 \lambda_{K} =  f_1 + c_0 f_0 + c_1 f_0 = 0\\
  \lambda_{K+1} + c_0 \lambda_{K} + c_0 \lambda_{K-1} =  f_0 + c_0 f_0 + c_1 f_1 = 0
 \end{array}
 \right.,
\end{equation}
which give $c_0 = 1$ and $c_1 =-2(1+Z)$ as the unique solutions.

Let us note that one can find the expression for $f_m$ in a non-recursive form, with the aid of the Z-transform:
\begin{equation}
 f_{m+2}+c_1 f_{m+1}+c_0 f_{m} = 0 \leftrightarrow F(z) = \frac{z(1+Z+c_1+z)}{z^2+z c_1 + c_0}.
\end{equation}
By undoing the Z-transform, we obtain the explicit expression for $f_m$:
\begin{equation}
\label{eq:fm}
 f_m = \frac{\alpha - 1}{\alpha - \beta} \alpha^m - \frac{\beta - 1}{\alpha - \beta} \beta^m,
\end{equation}
where $\alpha := (2 + Z + \sqrt{Z(4+Z)})/2$ and $\beta := (2 + Z - \sqrt{Z(4+Z)})/2$.

\begin{lemma}
\label{lem:Bnm}
Let $\lambda_m$ be defined as in \eqref{eq:fm}. The blocks $B_n^{(m)}$ are positive semidefinite and their rank is 2. Therefore, $\rho$ is PPT and its ranks are $(N+1, 2N....2N, 2N-1)$.
\end{lemma}
See the proof in Appendix \ref{proof:Bnm}.

\begin{lemma}
\label{lem:trace}
The trace of $\rho$ is $2(4+Z)^K$.
\end{lemma}
See the proof in Appendix \ref{proof:trace}.

\begin{lemma}
\label{lem:extremal}
The state $\rho$ is extremal in the PPT set.
\end{lemma}

\paragraph{Proof of Lemma \ref{lem:extremal}.} --
To prove extremality, we use the following theorem from \cite{AugusiakOpticsComm}: $\rho$ is extremal in the PPT set if, and only if, every Hermitian matrix $H$ satisfying ${\cal R}(H^{\Gamma_m})\subseteq {\cal R}(\rho^{\Gamma_m})$ is proportional to $\rho$.
Note that by taking $H\propto \rho$ we always find a solution satisfying the above conditions, but we have to show that no other exists.
Let us consider the subspace $E$ of the $(N+1)\times (N+1)$ Hermitian matrices spanned by the following matrices:
\begin{equation}
E=\left[\ket{D_N^0}\bra{D_N^0}, \ldots, \ket{D_N^N}\bra{D_N^N}, \frac{1}{\sqrt{2}}(\ket{D_N^N}\bra{D_N^0}+\ket{D_N^0}\bra{D_N^N})\right].
\end{equation}

Let us argue that we can assume that $H$ has to live in the same subspace as $\rho$. Since ${\cal R}(\rho)\subseteq E$, $H$ has to be of the form
\begin{equation}
H=\sum_{k=0}^{N} h_k \ket{D_{N}^k}\bra{D_{N}^k} + h\frac{1}{\sqrt{2}}(\ket{D_N^N}\bra{D_N^0}+\ket{D_N^0}\bra{D_N^N}),
\end{equation}
where $h_k$ and $h$ are real parameters.

The condition ${\cal R}(H^{\Gamma_m})\subseteq {\cal R}(\rho^{\Gamma_m})$ means that the vectors spanning $H^{\Gamma_m}$ must be orthogonal to (at least) the vectors in the kernel of $\rho^{\Gamma_m}$.
Fortunately, we have calculated the block-decomposition of $\rho^{\Gamma_m}$:
 \begin{equation}
  \rho^{\Gamma_m}=\left(
  \begin{array}{cc}
    \lambda_m & \sigma\\
    \sigma & \lambda_{N-m}
  \end{array}
  \right)\bigoplus \left(
    \bigoplus_{n=-m+1}^{N-m-1} A_{n}^{(m)}
    \right),
 \end{equation}

 As a side-comment let us observe that $(D_n^{(m)})^{-1}\ket{v}\in {\ker A_n^{(m)}}\iff \ket{v} \in {\ker B_n^{(m)}}$.
 Anyway, being orthogonal to $\ker{B_n}^{(m)}$ implies that the coefficients $h_m$ must satisfy a recurrence relation of the form
 \begin{equation}
 h_{m+2} - (2+Z) h_{m+1} + h_m = 0,
 \end{equation}
 which fixes $h_m \propto f_m$. Finally, we fix the value of $h$ by looking at the kernel of the block that goes alone in $\rho^{\Gamma_K}$, which is spanned by $(\sigma, -1)^T$. Hence, we have that $h_K \sigma - h = 0$, which implies that $h=\sigma h_K$. Hence, $H = h_K \rho \propto \rho$ is the only solution to ${\cal R}(H^{\Gamma_m})\subseteq {\cal R}(\rho^{\Gamma_m})$, proving that $\rho$ is extremal.
\qed

 Since all the states in the PPT set which are separable and extremal have ranks $r(\rho^{\Gamma_m})=1$, an extremal PPT state with a rank $r(\rho^{\Gamma_m})>1$ cannot be separable. Hence, $\rho(Z)$ is a uni-parametric family of PPT-entangled states for all $Z \in (0, \infty)$.

\section{Conclusions and Outlook}
\label{sec:Conclusions}
In this work we have studied the separability problem for diagonal symmetric states that are positive under partial transpositions. In the bipartite case, we have explored its connection to quadratic conic optimization problems, which naturally appear in a plethora of situations. Via this equivalence, we have been able to translate results from quantum information to optimization and vice-versa. For instance, we have characterized the extremal states of the set of separable DS states, defined entanglement witnesses for PPTDS states in terms of copositive matrices and we have rediscovered that the separability problem is NP-hard even in this highly symmetric and simplified case. We have shown that PPT is equivalent to separability in this context only for states of physical dimension not greater than $4$. We have complemented our findings with a series of analytical examples and counterexamples.
Furthermore, the state of the art in quadratic conic optimization allows us to see which are going to be the forthcoming challenges, in which insights developed within the quantum information community might contribute in advancing the field.

Second, we have provided a set of tools to certify separability of a bipartite PPTDS state in arbitrary dimensions, by decomposing it as a combination of an extremal DS state and a diagonal dominant DS state.. A natural further research direction is to study whether more sophisticated decompositions are possible, in terms of various extremal elements in ${\cal CP}_d$ and by understanding how the facial structure of ${\cal CP}_d$ plays a role in this problem.

Third, we have shown that, although $N$-qubit DS states are separable if, and only if, they are PPT with respect to every bipartition, just adding a new GHZ-like coherence can entangle the state while mantaining the PPT property for every bipartition. We have characterized analytically this class and we have shown that its ranks are much lower than those typically found in previous numerical studies \cite{AugusiakPRA12}.
In this search, we have also found an analytical example of a $4$-qubit PPT-entangled symmetric state, whose density matrix is sparse with real entries, when expressed in the computational basis, contrary to previous numerical examples \cite{TuraPRA12}. A natural following research direction is to connect the recently developed concept of coherence \cite{Coherence} to the Dicke basis, and to explore further its connection to PPT-entangled symmetric states.
Furthremore, it seems plausible that one can find further connections between the properties of $M(\rho)$ and those of MUBs since it has been shown that one can construct EWs based on MUBs that are capable of detecting bound entangled states (See \cite{EWfromMUBs} for a recent construction or \cite{BoundEntanglementMUBs} for an application to magic simplex states in an experimentally feasible way). Whether there is a clear connection between EWs from MUBs and EWs for bound entangled DS states in terms of the properties of $M(\rho)$ remains an open research direction.

\acknowledgements
We acknowledge financial support from the Spanish MINECO (SEVERO OCHOA grant SEV-2015-0522, FISICATEAMO FIS2016-79508-P, FIS2013-40627-P and FIS2016-86681-P AEI/FEDER EU), ERC AdG OSYRIS (ERC-2013-ADG No. 339106),  Generalitat de Catalunya (2014-SGR-874, 2014-SGR-966 and the CERCA Programme), and Fundació Privada Cellex. This project has received funding from the European Union’s Horizon 2020 research and innovation programme under the Marie Sk{\l}odowska-Curie grant agreement No 748549. R. Q. acknowledges the Spanish MECD for the FPU Fellowship (No. FPU12/03323). We thank S. Gharibian and two anonymous reviewers for insightful comments on the manuscript.

\nocite{apsrev41Control}
\bibliographystyle{apsrev4-1}
\bibliography{bibliography}

\appendix
\section{Proof of Theorem \ref{thm:Winterfell}}
\label{app:ProofOfWinterfell}

\begin{proof}
 Let us assume that $\rho$ is separable. Since it is symmetric, it admits a convex decomposition into product vectors of the following form:
 \begin{equation}
  \rho = \sum_{i} \lambda_{i} \ket{e_i}\ket{e_i}\bra{e_i}\bra{e_i},
  \label{eq:TheImp}
 \end{equation}
 where $\lambda_i$ form a convex combination and $\ket{e_i}:=\sum_{j=0}^{d-1} e_{ij}\ket{j}$, $e_{ij} \in \mathbbm{C}$.
 It follows that we have the identity
 \begin{equation}
  \rho=\sum_{i,x_1,x_2,y_1,y_2} \lambda_i e_{i,x_1} e_{i,x_2} e_{j,y_1}^* e_{j,y_2}^* \ket{x_1}\ket{x_2}\bra{y_1}\bra{y_2}=\sum_{0\leq a\leq b <d}p_{ab} \ket{D_{ab}}\bra{D_{ab}}.
  \label{eq:NedStark}
 \end{equation}
By projecting Eq. (\ref{eq:NedStark}) onto the Dicke basis we obtain the following conditions\footnote{There are, of course, more conditions that follow from Eq. (\ref{eq:NedStark}), such as $\sum_{i}\lambda_i (e_{ir})^2 (e_{is}^*)^2 = 0$, but we do not need them for this implication.}:
\begin{eqnarray}
 \bra{rr} \rho \ket{rr} &=& p_{rr} = \sum_{i}\lambda_i |e_{ir}|^4\nonumber\\
 \bra{D_{rs}} \rho \ket{D_{rs}} &=& p_{rs} = \sum_i \lambda_i 2 |e_{ir}|^2 |e_{is}|^2.
\end{eqnarray}
We can now construct $M(\rho)$, which has the form
\begin{equation}
 M(\rho)=\sum_{i} \lambda_i \left(
 \begin{array}{cccc}
  |e_{i0}|^4 & |e_{i0}|^2|e_{i1}|^2 & \cdots & |e_{i,0}|^2|e_{i,{d-1}}|^2\\
  |e_{i0}|^2|e_{i1}|^2 & |e_{i1}|^4 & \cdots & |e_{i,1}|^2|e_{i,{d-1}}|^2\\
  \vdots & \vdots & \ddots & \vdots\\
  |e_{i0}|^2|e_{i,{d-1}}|^2 & |e_{i1}|^2|e_{i,{d-1}}|^2 & \cdots & |e_{i,{d-1}}|^4
 \end{array}
 \right).
 \label{eq:SansaStark}
\end{equation}
It is clear from Eq. (\ref{eq:SansaStark}) that $M(\rho)$ is CP, as it admits a factorization $M(\rho)= \sum_i \vec{b_i} \cdot \vec{b_i}^T$, where $\vec{b_i}$ is a vector with components
\begin{equation}
 \vec{b}_i := \lambda_i^{1/2}\left(
 \begin{array}{cccc}
  |e_{i0}|^2 & |e_{i1}|^2 & \cdots & |e_{i,d-1}|^2
 \end{array}
 \right)^T.
 \label{eq:CerseiLannister}
\end{equation}
Clearly, we see from Eq. (\ref{eq:CerseiLannister}) that $M(\rho)$ is a convex combination of CP matrices, as $b_{ij} \geq 0$. Since CP matrices form a convex cone, $M(\rho)$ is CP. Actually, we can write $M(\rho)=B\cdot B^T$, where $\vec{b_i}$ are the columns of $B$.

Conversely, let us assume that $M(\rho)$ is CP. Note that, as $\rho$ is DS, $M(\rho)$ is in one-to-one correspondence with $\rho$. Since $M(\rho)$ is CP, we can write $M(\rho) = B \cdot B^T$, with $B$ being a $d\times k$ matrix fulfilling $B_{ij} \geq 0$. We have to give a separable convex combination of the form of Eq. (\ref{eq:TheImp}) that produces the DS $\rho$ matching the given $M(\rho)$. As we shall see, this separable decomposition is by no means unique. We begin by writing
\begin{equation}
 M(\rho)= \sum_{i=1}^k \vec{b_i} \cdot \vec{b_i}^T,
\end{equation}
where $\vec{b_i}$ are the columns of $B$, so all the coordinates of $\vec{b_i}$ are non-negative. Let $\{z_{ij}\}_{ij}$ be a set of complex numbers satisfying $|z_{ij}|^2= (\vec{b_i})_{j} \geq 0$ and let us define
\begin{equation}
 \ket{z_i} := \sum_{0\leq j< d}z_{ij} \ket{j}.
 \label{eq:TheMountain}
\end{equation}
Note that if we naively make the convex combination Eq. (\ref{eq:TheImp}) with the vectors introduced in Eq. (\ref{eq:TheMountain}), we shall produce a state with the corresponding $M(\rho)$, but it will not be DS in general. In order to ensure that the $\rho$ we are going to construct is DS we have to build it in a way that we eliminate all unwanted coherences. To this end, let us consider the more general family of vectors
\begin{equation}
 \ket{\zeta_{i,j,\mathbf{k}}}:= \sum_{0\leq l < d} (-1)^{k_l} \omega^{jl} z_{il} \ket{l}, \qquad 0 \leq j < d, \quad 0 \leq \mathbf{k} < 2^d,
\end{equation}
where $k_l$ is the $l$-th digit of $\mathbf{k}$ in base $2$ and $\omega$ is a primitive $2d-$th root of the unity (for instance, $\omega = \exp(2\pi \mathbbm{i}/2d)$). Let us now construct the (unnormalized) quantum state
\begin{equation}
 \rho_i = \sum_{0 \leq j < d} \sum_{0 \leq \mathbf{k} < 2^d} \ket{\zeta_{i,j,\mathbf{k}}} \ket{\zeta_{i,j,\mathbf{k}}}\bra{\zeta_{i,j,\mathbf{k}}} \bra{\zeta_{i,j,\mathbf{k}}}.
 \label{eq:TheHound}
\end{equation}
By expanding Eq. (\ref{eq:TheHound}) we obtain
\begin{equation}
 \rho_i = \sum_{j,\mathbf{k},l_1,l_2,l_3,l_4} (-1)^{k_{l_1}+ k_{l_2} + k_{l_3} + k_{l_4}} \omega^{j(l_1+l_2-l_3-l_4)}z_{i,l_1}z_{i,l_2} z_{i,l_3}^* z_{i,l_4}^* \ket{l_1, l_2} \bra{l_3, l_4},
\end{equation}
which we can rewrite as
\begin{equation}
 \rho_i = \sum_{0 \leq l_1,l_2,l_3,l_4<d} z_{i,l_1}z_{i,l_2} z_{i,l_3}^* z_{i,l_4}^* \ket{l_1, l_2} \bra{l_3, l_4} \left(\sum_{0 \leq \mathbf{k} < 2^d}(-1)^{k_{l_1}+ k_{l_2} + k_{l_3} + k_{l_4}}\right)\left(\sum_{0\leq j < d} \omega^{j(l_1+l_2-l_3-l_4)}\right).
 \label{eq:JonSnow}
\end{equation}
Let us inspect the possible values for the sums in parenthesis in Eq. (\ref{eq:JonSnow}). The expression involving $\mathbf{k}$ will be zero whenever $l_1$, $l_2$, $l_3$, $l_4$ are all different (half of the sum will come with a plus sign and the other half with a minus sign). If only two of them are equal, the expression is still zero by the same argument. If two of the indices are equal to the other two, then the value of the expression is $2^d$. Hence, we have that
\begin{equation}
\sum_{0 \leq \mathbf{k} < 2^d} (-1)^{k_{l_1} + k_{l_2} + k_{l_3} + k_{l_4}} = 2^d \left(\delta_{l_1-l_2}\delta_{l_3-l_4} + \delta_{l_1-l_3}\delta_{l_2-l_4} + \delta_{l_1-l_4}\delta_{l_2-l_3} - 2 \delta_{l_1-l_2}\delta_{l_2-l_3}\delta_{l_3-l_4}\right),
\label{eq:Pixie}
\end{equation}
where $\delta_x$ is the Kronecker Delta function (note that the last term prevents that the case $l_1=l_2=l_3=l_4$ is counted more than once). The second parenthesis in Eq. (\ref{eq:JonSnow}) is a geometrical series, so we have that it is $d$ if, and only if, $l_1 + l_2 \equiv l_3 + l_4 \mod 2d$ (because $\omega$ is taken to be primitive); otherwise it is $0$. As $0 \leq l_1+l_2, l_3+l_4 < 2d$, this can only happen if $l_1+l_2 = l_3+l_4$. Thus,
\begin{equation}
 \sum_{0 \leq j < d} \omega^{(l_1+l_2-l_3-l_4)j} = d \delta_{l_1+l_2-(l_3+l_4)}.
 \label{eq:Dixie}
\end{equation}
By inserting Eqs. (\ref{eq:Pixie}, \ref{eq:Dixie}) into Eq. (\ref{eq:JonSnow}), we have that the only possible values for $(l_1,l_2,l_3,l_4)$ are $(l_1,l_1,l_1,l_1)$, $(l_1,l_2,l_1,l_2)$ and $(l_1,l_2,l_2,l_1)$ (with $l_1 \neq l_2$). Note that Eq. (\ref{eq:Dixie}) forbids the combination $(l_1,l_1,l_2,l_2)$ if $l_1 \neq l_2$. This leads to
\begin{equation}
 \rho_i = d2^d\sum_{0\leq l_1 \neq l_2 < d} |z_{i,l_1}|^2 |z_{i,l_2}|^2 (\ket{l_1,l_2} \bra{l_1,l_2} +\ket{l_1,l_2} \bra{l_2,l_1}) + d2^d \sum_{0 \leq l_1 < d} |z_{i,l_1}|^4 \ket{l_1,l_1}\bra{l_1,l_1}.
 \label{eq:Kingslayer}
\end{equation}
By expressing Eq. (\ref{eq:Kingslayer}) in the Dicke basis, we have that $\rho_i$ is DS:
\begin{equation}
 \rho_i = d 2^d \sum_{0 \leq x < y < d} 2 |z_{i,x}| ^2 |z_{i,y}| ^2 \ket{D_{xy}}\bra{D_{xy}} + d2^d \sum_{0 \leq x < d} |z_{i,x}|^4 \ket{D_{xx}}\bra{D_{xx}}.
 \label{eq:KingTommen}
\end{equation}
Eq. (\ref{eq:KingTommen}) implies that $M(\rho_i)$ is
\begin{equation}
 M(\rho_i) = d2^d \vec{b_i} \cdot \vec{b_i}^T.
 \label{eq:HighSparrow}
\end{equation}
Therefore, the convex combination that we seek is
\begin{equation}
 \rho = \frac{1}{d2^d ||M(\rho)||_1}\sum_{0 \leq j < d} \sum_{0 \leq \mathbf{k} < 2^d} \ket{\zeta_{i,j,\mathbf{k}}}^{\otimes 2} \bra{\zeta_{i,j,\mathbf{k}}}^{\otimes 2},
 \label{eq:MadKing}
\end{equation}
where $||.||_1$ is the entry-wise $1$-norm (the sum of the absolute values of all the matrix entries). If $M(\rho)$ comes from a quantum state, then $||M(\rho)||_1=1$. Eq. (\ref{eq:MadKing}) proves that the state $\rho$ corresponding to $M(\rho)$ is separable.

\end{proof}

\section{Exposedness}
\label{app:Exposedness}
Convex sets are completely determined by their extremal elements (those that cannot be written as a proper convex combination of other elements in the set). An important step in characterizing the extremal elements of convex sets is understanding their facial structure. 
\begin{defn}
Given a convex cone $\cal K$, a face of $\cal K$ is a subset ${\cal F}\subseteq {\cal K}$ such that every line segment in the cone with an interior point in $\cal F$ must have both endpoints in $\cal F$.
\end{defn}
Note that every extreme ray of $\cal K$ is a one-dimensional face. To understand the facial structure of cones, one is interested in learning whether $\cal K$ is facially exposed. Facial exposedness is an important property that is exploited in optimization, allowing to design facial reduction algorithms \cite{FacialReduction}.
\begin{defn}
Let $\cal K$ be a cone in the space of real, symmetric matrices and let ${\cal F} \subseteq {\cal K}$ be a non-empty face. ${\cal F}$ is defined as an exposed face of $\cal K$ if, and only if, there exists a non-zero real symmetric matrix $A$ such that
\begin{equation}
{\cal K} \subseteq \{X \text{ s. t. } X\in M_{\mathbbm{R}}(d,d),\ X = X^T,\ \langle A, X \rangle \geq 0\}
\end{equation}
and
\begin{equation}
{\cal F} = \{X \in {\cal K} \text{ s. t. } \langle A, X \rangle = 0\}.
\end{equation}
\end{defn}
Hence, a face is exposed if it is the intersection of the cone with a non-trivial supporting hyperplane.

A cone is facially exposed if all of its faces are exposed. Although every extreme ray of ${\cal CP}_d$ is exposed \cite{dickinson2011geometry}, it remains unknown whether ${\cal CP}_d$ is facially exposed. In the case of ${\cal COP}_d$, the extreme rays corresponding to $\ket{ii}\bra{ii}$ are not exposed \cite{dickinson2011geometry}, implying that ${\cal PSD}_d + {\cal N}_d$ (the set of DEWs for PPTDS states) is not facially exposed. However, the set ${\cal DNN}_d$ of PPTDS states is facially exposed, because both ${\cal PSD}_d$ and ${\cal N}_d$ are facially exposed \cite{pataki2000geometry} and the intersection of facially exposed cones is facially exposed.

\section{Examples and counterexamples}
\label{sec:Examples}
\subsection{Every PPTDS state acting on ${\mathbbm C}^3\otimes {\mathbbm C}^3$ is separable}
\label{ex:PPT3}
In this example, we prove that every PPTDS state $\rho$ acting on ${\mathbbm C}^3\otimes {\mathbbm C}^3$ is separable. This follows from Theorem \ref{thm:TheEyre}, which is usually proven \cite{KoreanGuy} invoking results from quadratic non-convex optimization \cite{berman2015open}. We prove it here using quantum information tools solely: by building a convex separable decomposition of $\rho$ of the form of Eq. (\ref{eq:sep}). We do this in two steps. First, we provide a three-parameter class of PPTDS states that are separable. Then, by performing a Cholesky decomposition of $\rho^\Gamma$ we see that $\rho$ can be expressed as a convex combination of the family we introduced, for some parameters. The PPT conditions directly relate to the existence of such a Cholesky decomposition.

Recall that $\rho$ is written as
\begin{equation}
 \rho = \sum_{0\leq i \leq j <3} p_{ij} \ket{D_{ij}}\bra{D_{ij}},
 \label{eq:defRhoExample3}
\end{equation}
where
$\ket{D_{ii}} = \ket{ii}$ and $\ket{D_{ij}} = (\ket{ij} + \ket{ji})/\sqrt{2}$ if $i < j$.
Short algebra shows that $\rho$ and its partial transpose $\rho^\Gamma$ have the form
\begin{equation}
 \rho = \bigoplus_{0\leq i \leq j < 3}\left( p_{ij}\right),
\end{equation}
and
\begin{equation}
 \rho^\Gamma = \left(\frac{p_{01}}{2}\right)\oplus\left(\frac{p_{01}}{2}\right)\oplus\left(\frac{p_{02}}{2}\right)\oplus\left(\frac{p_{02}}{2}\right)\oplus\left(\frac{p_{12}}{2}\right)\oplus\left(\frac{p_{12}}{2}\right)\oplus M,
\end{equation}
where
\begin{equation}
 M = \left(
 \begin{array}{ccc}
  p_{00} & p_{01}/2 & p_{02}/2\\
  p_{01}/2 & p_{11} & p_{12}/2\\
  p_{02}/2 & p_{12}/2 & p_{22}
 \end{array}
 \right).
\end{equation}

\begin{lemma}
\label{lem:ex:1}
Let $x,y,z \in {\mathbbm{C}}$. Let $\omega$ be a primitive third root of the unity: $\omega^3=1$. Let us denote $P_{x,y,z} := \ket{\psi_{x,y,z}}\bra{\psi_{x,y,z}}$, where $\ket{\psi_{x,y,z}}:=x\ket{0}+y\ket{1}+z\ket{2}$ (we do not normalize $\ket{\psi_{x,y,z}}$). Let us further define
\begin{equation}
Q_{x,y,z} := P_{x,y,z}^{\otimes 2} + P_{x,\omega y, \omega^2 z}^{\otimes 2} + P_{x,\omega^2 y, \omega z}^{\otimes 2}.
\end{equation}
 Then, the unnormalized quantum state
\begin{equation}
\sigma_{x,y,z}:=\frac{1}{12}\left(Q_{x,y,z} + Q_{-x,y,z} + Q_{x,-y,z} + Q_{x,y,-z}\right)
\end{equation}
is diagonal symmetric. (Obviously it is PPT, as it is separable). Furthermore, its expression in Eq. (\ref{eq:defRhoExample3}) corresponds to
\begin{equation}
\left\{
\begin{array}{lll}
p_{00} &=& |x|^4\\
p_{01} &=& 2|x|^2|y|^2\\
p_{02} &=& 2|x|^2|z|^2\\
p_{11} &=& |y|^4\\
p_{12} &=& 2|y|^2|z|^2\\
p_{22} &=& |z|^4
\end{array},
\right.
\end{equation}
where $|\cdot|$ denotes the complex modulus.
\end{lemma}

\begin{proof} The proof follows from expressing $\sigma_{x,y,z}$ in the computational basis. After some elementary algebra, one arrives at the form of Eq. (\ref{eq:defRhoExample3}).

\end{proof}

The following Lemma allows us to find a decomposition of a positive semi-definite matrix $A$ of the form $A = B \cdot B^T$. To this end, we apply Cholesky's decomposition.
\begin{lemma}
\label{lem:ex:2}
Let $A$ be a real, symmetric, positive-semidefinite $3\times 3$ matrix given by
\begin{equation}
A=\left(
\begin{array}{ccc}
a&b&c\\
b&d&e\\
c&e&f
\end{array}
\right).
\end{equation}
Then, $A$'s Cholesky decomposition can be written as
\begin{eqnarray}
A &=& \frac{1}{a}\left(a,b,c\right)^T\left(a,b,c\right) + \frac{1}{a\left|
\begin{array}{cc}
a&b\\b&d
\end{array}
\right|} \left(0,\left|
\begin{array}{cc}
a&b\\b&d
\end{array}
\right|, \left|
\begin{array}{cc}
a&c\\b&e
\end{array}
\right|\right)^T\left(0,\left|
\begin{array}{cc}
a&b\\b&d
\end{array}
\right|, \left|
\begin{array}{cc}
a&c\\b&e
\end{array}
\right|\right)\nonumber\\
&+& \frac{1}{\left|
\begin{array}{cc}
a&b\\b&d
\end{array}
\right|\det A} (0,0,\det A)^T (0,0, \det A).
\label{eq:Cholesky}
\end{eqnarray}
\end{lemma}

\begin{proof}
The idea of the proof is to use the rank-1 matrix $A_1:= (a,b,c)^T(a,b,c)/a$ to fix the elements of $A$ that lie on the first column and first row. Then, the second summand will adjust the elements of the second row, second column of $A$ and the last summand will fix the bottom-right element of $A$. Therefore, we have
\begin{equation}
A_1=\left(
\begin{array}{ccc}
a&b&c\\
b&\cdot&\cdot\\
c&\cdot&\cdot
\end{array}
\right),
\end{equation}
where the $\cdot$ are terms that are not yet fixed.
When we add the second term to $A_1$ we have
\begin{equation}
A_2=\left(
\begin{array}{ccc}
a&b&c\\
b&d&e\\
c&e&\cdot
\end{array}
\right),
\end{equation}
and adding the last term to $A_2$ we recover $A$.

\end{proof}

Now we have the tools to prove that every DNN $3\times 3$ matrix is CP:
\begin{lemma}
\label{lem:ex:3}
If $A$ is a $3\times 3$ positive-semidefinite matrix, and it is entry-wise non-negative, then there exists a Cholesky decomposition of $A$ with non-negative vectors (i.e., the vectors' coordinates are non-negative).
\end{lemma}

\begin{proof}
The only problematic term in Eq (\ref{eq:Cholesky}) is $\left|
\begin{array}{cc}
a&c\\b&e
\end{array}
\right|$ as all the other expressions are either principal minors of $A$ or entries of $A$, so they are non-negative.
Recall that the Cholesky decomposition of $A$ picks an order of rows, but this is arbitrary; it could be done in any order. Just to illustrate it, if we reorder the columns and rows of $A$ then 
we have that all the possibilities are

\begin{equation}
\left\{
\left(
\begin{array}{ccc}
a&b&c\\
b&d&e\\
c&e&f
\end{array}
\right),
\left(
\begin{array}{ccc}
a&c&b\\
c&f&e\\
b&e&d
\end{array}
\right),
\left(
\begin{array}{ccc}
d&b&e\\
b&a&c\\
e&c&f\\
\end{array}
\right),
\left(
\begin{array}{ccc}
d&e&b\\
e&f&c\\
b&c&a\\
\end{array}
\right),
\left(
\begin{array}{ccc}
f&c&e\\
c&a&b\\
e&b&d
\end{array}
\right),
\left(
\begin{array}{ccc}
f&e&c\\
e&d&b\\
c&b&a
\end{array}
\right)
\right\}.
\end{equation}

Hence, the corresponding minors that could be negative are
\begin{equation}
\left\{\left|
\begin{array}{cc}
a&c\\b&e
\end{array}
\right|,\left|
\begin{array}{cc}
d&e\\b&c
\end{array}
\right|,\left|
\begin{array}{cc}
f&e\\c&b
\end{array}
\right|\right\}.
\label{eq:ex:setofnonnegative}
\end{equation}

If any of the numbers in  (\ref{eq:ex:setofnonnegative}) is non-negative, then we can pick the Cholesky decomposition for that particular order and we obtain the result. The alternative is that all of them are strictly negative. We are going to see now that this would contradict the fact that $A \succeq 0$.

Note that all the numbers in (\ref{eq:ex:setofnonnegative}) being strictly negative imply that $d>0$. Otherwise, if $d=0$, since $A\succeq 0$, this would imply that $b=e=0$. Then, all the numbers in (\ref{eq:ex:setofnonnegative}) would be zero. Therefore, $d$ must be strictly positive. Similarly, $b>0$. Otherwise, if $b=0$, then we would have $c d < 0$ and $ae < 0$, contradicting the fact that $A$ is entry-wise non-negative. Therefore, $b>0$.

It is sufficient to find a contradiction just with a subset of the conditions given by (\ref{eq:ex:setofnonnegative}):
Let us assume that $ae<bc$ and $cd <be$. Then, we have that
\begin{equation}
aed < bcd < b^2e,
\end{equation}
where we used $ae<bc$ and $d>0$ in the first inequality and $cd < be$ and $b>0$ in the second. Therefore, $aed < b^2e$. Hence, it must be that $e>0$ (otherwise we would have $0<0$). Then, we deduce that $ad < b^2$, but this directly contradicts $A \succeq 0$, as the latter implies $ad \geq b^2$.

\end{proof}

Now we have the necessary tools to prove the result claimed in the example.

Consider $\rho$ to be a PPTDS state. Then, we have that $p_{ij}\geq 0$ and $M \succeq 0$. We want to write $\rho$ as a convex combination of some elements $\sigma_{x,y,z}$ introduced in Lemma \ref{lem:ex:1} by appropriately picking $x,y,z$ as functions of $p_{ij}$. Observe that for all $x,y,z \in \mathbbm{C}$, the entries of $\sigma_{x,y,z}$ will be non-negative. Moreover, the matrix $M$ associated to $\sigma_{x,y,z}$ is
\begin{equation}
M_{\sigma}=\left(
\begin{array}{ccc}
|x|^4&|x|^2|y|^2&|x|^2|z|^2\\
|x|^2|y|^2&|y|^4&|y|^2|z|^2\\
|x|^2|z|^2&|y|^2|z|^2&|z|^4
\end{array}
\right),
\end{equation}
which has rank $1$, and it is generated as $ M_{\sigma} = (|x|^2,|y|^2,|z|^2)^T(|x|^2,|y|^2,|z|^2)$.
Hence, the idea is to relate $M_\sigma$ to each element of the Cholesky decomposition Eq. (\ref{eq:Cholesky}) so that their sum gives the original $M$. If we recover the given $M$, we automatically recover $\rho$ and we have a separable convex decomposition of it. This can be done if, and only if, the components of the vectors appearing in Eq. (\ref{eq:Cholesky}) are non-negative, because we can always find numbers $x,y,z \in \mathbbm{C}$ realizing them.

Let us then apply Lemma \ref{lem:ex:2} to $M_\sigma$: We want to generate $\rho= \lambda_0\sigma_{x_0, y_0,z_0} + \lambda_1\sigma_{x_1, y_1, z_1} + \lambda_2\sigma_{x_2, y_2, z_2}$.
We begin by picking
$$\lambda_0 = \frac{1}{p_{00}}, \qquad (x_0,y_0,z_0)=(\sqrt{p_{00}},\sqrt{p_{01}/2},\sqrt{p_{02}/2}).$$ All the components are non-negative by hypothesis. Let us now move to $(x_1,y_1,z_1)$. We now pick $$\lambda_{1}= \frac{1}{p_{00}(p_{00} p_{11} -p_{01}^2/4)}, \qquad (x_1,y_1,z_1)=\left(0,\sqrt{p_{00} p_{11} -p_{01}^2/4}, \sqrt{p_{00}p_{12}/2-p_{01}p_{02}/4}\right).$$
In this case $p_{00} p_{11} -p_{01}^2/4\geq 0$ because it is a principal minor of $M$. So, $\lambda_1 \geq 0$ and $y_1 \geq 0$. However, $z_1$ might need to be negative. We deal with this case at the end.

Finally, we consider $(x_2,y_2,z_2)$. Now we have
$$\lambda_{2}= \frac{1}{(p_{00} p_{11} -p_{01}^2/4)\det M}, \qquad (x_2,y_2,z_2)=\left(0,0, \sqrt{\det M}\right).$$
Here it is easy to see that $\lambda_2 \geq 0$ and $z_2 \geq 0$.

The proof is finished if we can argue that $z_1$ can be taken to be a positive number. This is guaranteed by Lemma \ref{lem:ex:3}, which tells us that there exists always a relabeling of the computational basis elements $\ket{0}$, $\ket{1}$ and $\ket{2}$ such that the Cholesky decomposition of $M$ is done with non-negative vectors.

\subsection{An example of a PPTDS entangled state acting on $\mathbbm{C}^6 \otimes \mathbbm{C}^6$.}
\label{ex:d6}
We now present an example for $d=6$ that constitutes an unnormalized PPTDS entangled state. This is based on a counterexample that appeared in the context of financial engineering \cite{Sonneveld}.

Let $p_{ii} = 2, p_{i,i+1}=3, p_{i,i+2}=1, p_{i,i+3}=0, p_{i,i+4}=1, p_{i,i+5}=3$. This means that the matrix $M$ takes the form
\begin{equation}
M=\left(
\begin{array}{cccccc}
2&3/2&1/2&0&1/2&3/2\\
3/2&2&3/2&1/2&0&1/2\\
1/2&3/2&2&3/2&1/2&0\\
0&1/2&3/2&2&3/2&1/2\\
1/2&0&1/2&3/2&2&3/2\\
3/2&1/2&0&1/2&3/2&2
\end{array}
\right).
\end{equation}
Observe that $M$ is a circulant matrix. More importantly, $M$ factorizes as $M=Z^T \cdot Z$, where
\begin{equation}
Z=\left(
\begin{array}{cccccc}
1&1&1&1&1&1\\
0&\sqrt{3}/2&\sqrt{3}/2&0&-\sqrt{3}/2&-\sqrt{3}/2\\
1&1/2&-1/2&-1&-1/2&1/2
\end{array}
\right).
\end{equation}
Note that this proves that $\rho$ is PPT. The matrix $M$ does not admit a non-negative matrix factorization \cite{Sonneveld}, so we cannot apply the separable decomposition of the $3\otimes 3$ case.

The matrix $M$ has rank $3$ and its kernel is given by three vectors orthogonal to $Z$, which are
\begin{equation}
\left(
\begin{array}{cccccc}
1&0&0&-1&2&-2\\
0&1&0&-2&3&-2\\
0&0&1&-2&2&-1
\end{array}
\right).
\end{equation}

Now we can apply the range criterion to $\rho$. So, if $\rho$ is separable, there will exist $\ket{\psi}=\ket{\zeta}_A\ket{\zeta}_B$ in the range of $\rho$ such that $\ket{\psi^c} = \ket{\zeta}_A\ket{\zeta^*}_B$ is in the range of $\rho^\Gamma$ (I assume the same vector $\ket{\zeta}$ on $A$ and $B$ as $\rho$ acts on the symmetric space). As a vector belongs to the range iff it is orthogonal to the kernel, the range criterion implies that, if $\rho$ is separable, then the system of equations imposed by $\ket{\psi}\bot \ker \rho$ and $\ket{\psi^c} \bot \ker{\rho^\Gamma}$ has a non-trivial solution.

Let us parametrize $\ket{\zeta} = \sum_{0\leq i < 6} z_i \ket{i}$. Then we have that
\begin{equation}
\ket{\psi}= \sum_{0\leq i , j < 6} z_i z_j \ket{ij}
\end{equation}
and
\begin{equation}
\ket{\psi^c}= \sum_{0\leq i , j < 6} z_i z_j^* \ket{ij}.
\end{equation}

The kernel of $\rho$ is spanned by $\ket{D_{03}}, \ket{D_{14}}$ and $\ket{D_{25}}$, because $p_{i,i+3}=0$.
This gives the equations
\begin{equation}
z_0z_3 = z_1 z_4 = z_2 z_5 = 0.
\end{equation}
The kernel of $\rho^\Gamma$ is given by the vectors $\ket{i,i+3}$ and $\ket{i+3,3}$, and the vectors in the kernel of $M$, in the appropriate basis. The first ones introduce redundant equations
\begin{equation}
z_0z_3^* = z_1 z_4^* = z_2 z_5^* = 0.
\end{equation}

So all the important information comes from the kernel of $M$. This means that
\begin{eqnarray}
(\bra{00} - \bra{33} +2 \bra{44} -2 \bra{55})\ket{\psi^c}=0\nonumber\\
(\bra{11} - 2\bra{33} +3 \bra{44} -2 \bra{55})\ket{\psi^c}=0\nonumber\\
(\bra{22} - 2\bra{33} +2 \bra{44} - \bra{55})\ket{\psi^c}=0\nonumber\\
\end{eqnarray}

It follows that the above system can be compacted as
\begin{equation}
\left(
\begin{array}{c}
|z_0|^2\\|z_1|^2\\|z_2|^2
\end{array}
\right) = \left(
\begin{array}{ccc}
1&-2&2\\2&-3&2\\2&-2&1
\end{array}
\right) \left(
\begin{array}{c}
|z_3|^2\\|z_4|^2\\|z_5|^2
\end{array}\right).
\label{eq:3x3}
\end{equation}

There are a few cases to consider now, to include the conditions $z_0z_3=z_1z_4=z_2z_5=0$:
\begin{itemize}
\item If $z_3=z_4=z_5=0$, then the above system implies $z_0=z_1=z_2=0$.
\item Conversely, as the $3\times 3$ matrix in Eq. (\ref{eq:3x3}) is invertible (actually, it is its own inverse), if $z_0=z_1=z_2=0$, then the above system implies  $z_3=z_4=z_5=0$.
\item If two of the numbers in $\{z_3,z_4,z_5\}$ are zero (then one number in $\{z_0,z_1,z_2\}$ is also zero) we have that the system becomes of the following form (for instance, for the case $z_0=z_4=z_5=0$)
\begin{equation}
\left(
\begin{array}{c}
0\\|z_1|^2\\|z_2|^2
\end{array}
\right) = |z_3|^2\left(
\begin{array}{c}
1\\2\\2
\end{array}
\right).
\end{equation}
This also implies that $z_3=0$, which implies that $z_1=z_2=0$.
\item The remaining case is that one of the numbers in $\{z_3,z_4,z_5\}$ are zero and two of the numbers of $\{z_0,z_1,z_2\}$ are zero. By inverting Eq. (\ref{eq:3x3}), this reduces to the previous case.
\end{itemize}
Hence, the only solution to the above system of equations is that
$z_0=z_1=z_2=z_3=z_4=z_5=0$. This does not give a valid quantum state. Hence, there does not exist a quantum state $\psi$ with the properties required by the range criterion. Consequently, $\rho$ is entangled.

\section{Proofs and examples for Section \ref{sec:sufficient}}

\subsection{Example of an entangled PPTDS state for $d=5$.}
\label{ex:5x5second}
Example for $d=5$ would be \cite{berman2003completely}
\begin{equation}
\hat{M}(\rho)=\left(
\begin{array}{ccccc}
1&1&0&0&1\\
1&2&1&0&0\\
0&1&2&1&0\\
0&0&1&2&1\\
1&0&0&1&6
\end{array}
\right).
\end{equation}
In this case, $\ker {\hat{M}(\rho)} = \{\vec{0}\}$, so to apply the Range criterion one needs to subtract some rank-$1$ projectors, which are $3/16 \vec{v_1}\vec{v_1}^T + 1/16 \vec{v_2}\vec{v_2}^T$, where $\vec{v_1}^T = (1,0,0,0,1)$ and $\vec{v_2}=(1,0,0,0,9)$.\\
Equivalently, we can use the Horn matrix as an EW to certify entanglement $\mathrm{Tr}(H (\hat{M}(\rho)-3/16 \vec{v_1}\vec{v_1}^T - 1/16 \vec{v_2}\vec{v_2}^T))=-1< 0$.

\subsection{Proof of Theorem \ref{thm:Anna}}
\label{proof:Annathm}
\begin{proof}
Let us rewrite $\rho = \tilde{\rho} + \varepsilon I$, where $\tilde{\rho} = \rho - \varepsilon I$.
Let us make the following observations:
\begin{enumerate}
\item $\tilde{\rho}$ is a legitimate DS matrix: $\tilde{\rho}_{ij} \geq 0$. This comes from the fact that $\tilde{\rho}_{ij} = \rho_{ij} - \varepsilon$ and the first hypothesis is precisely $\rho_{ij} - \varepsilon \geq 0$.
\item $\tilde{\rho}$ is PPT. Since $\tilde{\rho}_{ij} \geq 0$, the only remaining condition to prove is that $\tilde{M}(\rho) \succeq 0$. Note that $\tilde{M}(\rho) = M(\rho) - d \varepsilon \ket{u} \bra{u}$. We want to prove that, for any vector $\ket{v}$, we have that $\bra{v}(M - \varepsilon d \ket{u}\bra{u})\ket{v} \geq 0$. Note that, since $\ket{u} \in {\cal R}(M(\rho))$, there exists $\ket{\Psi}$ such that $\ket{u} = M(\rho)\ket{\Psi}$. Therefore, we can write $\braket{v}{u} \braket{u}{v} = |\bra{v} \sqrt{M(\rho)} \frac{1}{\sqrt{M(\rho)}}\ket{u}|^2$ and, by virtue of Cauchy-Schwarz inequality, $\braket{v}{u} \braket{u}{v} \leq \bra{v}M(\rho)\ket{v}\bra{u}\frac{1}{M(\rho)}\ket{u}$. Note that the positive semi-definiteness of $M(\rho)$ allows us to pick a square root such that $\sqrt{M(\rho)}\succeq 0$. Thus, we have that
\begin{equation}
\bra{v}M(\rho)\ket{v} \geq \braket{v}{u} \braket{u}{v} (\bra{u}\frac{1}{M(\rho)}\ket{u})^{-1} \geq \braket{v}{u} \braket{u}{v} d \varepsilon,
\end{equation}
which yields the result $\bra{v} \tilde{M}(\rho)\ket{v} \geq 0$ for all $\ket{v}$.

\item A sufficient condition for a real symmetric non-negative matrix to be completely positive is that it is diagonally dominant \cite{kaykobad1987nonnegative}. Therefore, if we prove that $\tilde{M}(\rho)$ is diagonally dominant, the associated $\tilde{\rho}$ will be a DS separable state. This is guaranteed by the third hypothesis:
$$\tilde{\rho}_{ii} = \rho_{ii} - \varepsilon \geq \sum_{j \neq i} \rho_{ji} - (d-1) \varepsilon = \sum_{j \neq i} \tilde{\rho}_{ji}.$$
Hence, $\rho$ is separable.
\end{enumerate}
\end{proof}

\subsection{An example for Theorem \ref{thm:Anna}}
\label{ex:Albert1}
Let us construct an example to show that ${\cal CP}_d \setminus {\cal DD}_d \neq \emptyset$ and then show how to guarantee separability using Theorem \ref{thm:Anna}.

Consider the case of a DS state $\rho$ that takes the form
\begin{equation}
	\rho = \sum_{i=0}^{d-1} \alpha \ket{ii}\bra{ii} + \sum_{0 \leq i < j < d} 2 \beta \ket{D_{ij}}\bra{D_{ij}}
	\label{eq:96}
\end{equation}
and coefficients $\alpha, \beta \in \mathbb{R}_{\geq 0}$. Normalization imposes the constraint $d \alpha + d(d-1)\beta=1$.\\
Let us choose $\alpha$ and $\beta$ in such a way that $M(\rho)$ lies in the line segment between $M(I)$ and $M(\tilde{\rho})$; \textit{i.e.}, it is a convex combination of the following form:
\begin{equation}\label{eq:example}
M(\rho)=\lambda M(\tilde{\rho})+(1-\lambda)M(I) \text{ for } 0\leq \lambda\leq 1,
\end{equation}
where $M(I)$ is an extremal of $\mathcal{CP}_d$, with the rank-1 state $I$ defined as in Lemma \ref{lem:I}, and $M(\tilde{\rho})$ has the same form as $M(\rho)$ but with coefficients $\tilde{\alpha},\tilde{\beta}$ chosen as $\tilde{\alpha}=(d-1)\tilde{\beta}$ which corresponds to the limit where $M(\rho)$ becomes $\mathcal{DD}_d$. Together with the normalization constraint, one obtains $\tilde{\alpha}=(2d)^{-1}$ and $\tilde{\beta}=(2d(d-1))^{-1}$. \\ 

Therefore, by construction, any choice of $\lambda\in \left[0,1 \right)$ will yield a state with associated $M(\rho)$ being  $\mathcal{CP}_d\setminus\mathcal{DD}_d$. \\

Now let's proceed to show how to guarantee separability of $\rho$ by virtue of Theorem \ref{thm:Anna} given a two-qudit PPTDS state $\rho$. Take, for instance, \eqref{eq:example} with $\lambda=1/2$ resulting in $M(\rho)=1/2(M(\tilde{\rho})+M(I))$ (which is in $\mathcal{CP}_d\setminus\mathcal{DD}_d$ by construction. To guarantee separability using Theorem \ref{thm:Anna}, we are going to find a decomposition $\rho=\tilde{\rho}+\epsilon I$ showing that there exists an $\varepsilon$ that fulfils the conditions of the theorem (we are going to assume $d\geq 5$).\\
Condition 1 gives an upper bound given by $\varepsilon \leq \min\{\alpha, \beta\}=\alpha=\frac{d+2}{4d^2}$.
Condition 2 gives a more restrictive upper bound
\begin{equation}
\varepsilon \leq \frac{1}{d}(\bra{u} \frac{1}{M(\rho)}\ket{u})^{-1}=\frac{\alpha+(d-1)\beta}{d}=1/d^2,
\end{equation}
where the pseudoinverse can be found via the Sherman-Morrison formula (because of the particular form of Eq. (\ref{eq:96})). Finally condition 3 gives the lower bound
\begin{equation}
\varepsilon\geq\frac{\beta(d-1)-\alpha}{(d-2)}=\frac{d-1}{4d^2(d-2)}.
\end{equation}
Therefore, we have certified the separability of $\rho$ since we can find the desired decomposition $\rho=\tilde{\rho}+\epsilon I$ for all $\varepsilon \in [\frac{d-1}{4d^2(d-2)}, \frac{1}{d^2}]$.

\subsection{Proof of Lemma \ref{lem:X}}
\label{proof:lemmaIx}

\begin{proof}
Let $\ket{e(\vec \varphi)} = \sum_{i=0}^{d-1} \sqrt{x_i/||\mathbf{x}||_1} e^{\mathbbm{i}\varphi_i} \ket{i}$. A separable decomposition of $I_{\mathbf x}$ is given by
\begin{equation}
I_{\mathbf{x}} = \int_{[0, 2\pi]^d} \frac{\mathrm{d}\vec{\varphi}}{(2\pi)^d}(\ket{e(\vec\varphi)}\bra{e(\vec{\varphi})})^{\otimes 2}.
\end{equation}
Indeed, note that
\begin{eqnarray}
 I_{\mathbf x} &=& \sum_{ijkl} \ket{ij}\bra{kl}\int_{[0, 2\pi]^d} \frac{\mathrm{d}\vec{\varphi}}{(2\pi)^d} \frac{\sqrt{x_i x_j x_k x_l}}{||\mathbf{x}||_1^2} e^{{\mathbbm i}(\varphi_i + \varphi_j - \varphi_k - \varphi_l)}\nonumber\\
 &=& \sum_{ijkl}\ket{ij}\bra{kl} \frac{\sqrt{x_i x_j x_k x_l}}{||\mathbf{x}||_1^2} (\delta_{i,k}\delta_{j,l}+ \delta_{i,l}\delta_{j,k} - \delta_{i,j,k,l}),
\end{eqnarray}
where $\delta$ is the Kronecker delta function.
\end{proof}

\subsection{Proof of Theorem \ref{thm:JordisGen}}
\label{proof:JordisGen}

\begin{proof}
We write $\rho = (1-\lambda)\tilde{\rho} + \lambda I_{\mathbf x}$. Therefore, we have that
\begin{equation}
 M(\tilde \rho) = \frac{1}{1-\lambda}(M(\rho) - \lambda M(I_{\mathbf{x}})).
\end{equation}
Our goal is to prove that $M(\tilde \rho)$ is completely positive. Therefore, $M(\rho)$ will also be completely positive and $\rho$ will be a separable quantum state.
\begin{enumerate}
 \item We start by showing that $M(\tilde \rho)$ is non-negative component-wise. Indeed, we have that
 \begin{equation}
  (M(\tilde \rho))_{ij} = \frac{1}{1-\lambda}((M(\rho))_{ij} - \lambda (M(I_{\mathbf x}))_{ij}) \geq \frac{1}{1-\lambda}((M(\rho))_{ij} - \frac{(M(\rho))_{ij}||\mathbf{x}||_1^2}{x_ix_j}(M(I_{\mathbf x}))_{ij}) = 0,
 \end{equation}
because $(M(I_{\mathbf x}))_{ij} = \frac{x_i x_j}{||\mathbf{x}||_1^2}$.
 \item Now we show that $M(\tilde \rho)$ is positive semi-definite. This means that $\bra{v}M(\tilde \rho)\ket{v} \geq 0$ for every $\ket{v}$.
 Since we assume that $1 - \lambda > 0$, it suffices to check that $\bra{v} (M(\rho) - \lambda M(I_{\mathbf{x}})) \ket{v} \geq 0$ holds.
 Recall that $\bra{v}M(I_{\mathbf{x}})\ket{v} = |\braket{v}{u_{\mathbf x}}|^2$. Since $\ket{u_{\mathbf x}} \in {\cal R}(M(\rho))$, it means that $\bra{u_{\mathbf x}} M(\rho) \ket{u_{\mathbf x}} > 0$ and therefore $\bra{u_{\mathbf x}} \frac{1}{M(\rho)} \ket{u_{\mathbf x}} > 0$. Therefore, we can apply the Cauchy-Schwarz inequality to $\bra{v}M(I_{\mathbf{x}})\ket{v}$ and obtain
 \begin{equation}
  \braket{v}{u_{\mathbf x}}\braket{u_{\mathbf x}}{v} = |\bra{v}\sqrt{M(\rho)} \frac{1}{\sqrt{M(\rho)}} \ket{u_{\mathbf{x}}}|^2 \leq \bra{v} M(\rho) \ket{v} \bra{u_{\mathbf x}}\frac{1}{M(\rho)} \ket{u_{\mathbf{x}}}.
 \end{equation}
Note that the positive semi-definiteness of $M(\rho)$ enables us to choose a square root branch of $M(\rho)$ such that $\sqrt{M(\rho)} \succeq 0$.
Therefore, we have that 
\begin{equation}
 \bra{v} (M(\rho) - \lambda M(I_{\mathbf{x}})) \ket{v} \geq \bra{v} M(\rho) \ket{v} (1 - \lambda \bra{u_{\mathbf x}}\frac{1}{M(\rho)} \ket{u_{\mathbf x}}) \geq \bra{v} M(\rho) \ket{v}(1 - 1) = 0.
\end{equation}
  \item At this point we have proved that conditions $1$ and $2$ of the theorem guarantee that $M(\tilde \rho)$ is doubly non-negative. The third condition will guarantee that it is diagonal dominant.
  In order to show
  \begin{equation}
   (M(\tilde \rho))_{ii} - \sum_{j \neq i} (M(\tilde \rho))_{ij} \geq 0
  \end{equation}
  for all $i$, we note that this can be rewritten as
  \begin{equation}
   (M(\rho))_{ii} - \sum_{j\neq i}(M(\rho))_{ij} - \lambda \left(\frac{x_i ^2}{||\mathbf{x}||_1^2} - \sum_{j \neq i}\frac{x_i x_j}{||\mathbf{x}||_1^2}\right) \geq 0.
  \end{equation}
By adding and subtracting $x_i^2$ to the parenthesis, we can rearrange the condition we want to prove as
  \begin{equation}
   (M(\rho))_{ii} - \sum_{j\neq i}(M(\rho))_{ij} - \frac{\lambda x_i}{||\mathbf{x}||_1^2} (2 x_i - ||\mathbf{x}||_1) \geq 0,
  \end{equation}
  which from this form, the result follows immediately from the fact that our assumption is $\lambda x_i (||\mathbf{x}||_1 - 2 x_i) \geq ||\mathbf{x}||_1^2 \left[ \sum_{j \neq i} (M(\rho))_{ij} - (M(\rho))_{ii}\right]$ for all $i$.
\end{enumerate}
Therefore, the conditions of the theorem guarantee that we have a diagonal dominant $M(\tilde \rho)$ that is also doubly non-negative. Since every real symmetric non-negative matrix that is diagonally dominant is completely positive \cite{kaykobad1987nonnegative}, the associated $\tilde \rho$ is a separable DS state. Therefore, $\rho$ can be written as a convex combination of separable states; therefore $\rho$ is separable.

\end{proof}

\section{Proofs and examples for Section \ref{sec:Class}}

\subsection{A four-qubit PPT-entangled symmetric state with a closed analytical form}
\label{ex:4qubits}

Following the procedure described in the beginning of Section \ref{sec:Class}, we have been able to find an example of a four-qubit PPT-entangled symmetric state.
\begin{equation}
\rho = \frac{1}{50\sqrt{7}}\left(
\begin{array}{ccccc}
7\sqrt{7}&&&&\\
&12\sqrt{7}&&&-2\sqrt{15}\\
&&12\sqrt{7}&&\\
&&&12\sqrt{7}&\\
&-2\sqrt{15}&&&7\sqrt{7}
\end{array}
\right).
\end{equation}
This state has the same typical ranks as the ones found numerically in \cite{TuraPRA12, AugusiakPRA12}; \textit{i.e.}, the rank of $\rho$ is $5$, the rank of $\rho^{T_{A}}$ is $7$ (there is one vector in its kernel) and the rank of $\rho^{T_{AB}}$ is $8$ (there is another vector in its kernel).

\subsection{Proof Lemma \ref{lem:blocks}}
\label{proof:blocks}
\begin{proof}
Let us start by recalling that a symmetric diagonal state $\rho$ with diagonal elements $\{\lambda_k'\}_{k=0}^N$ and GHZ coherences $o(\sigma)$ is expressed as follows, when acting on ${\mathbbm C}^{m+1}\otimes {\mathbbm C}^{n-m+1}$:
 \begin{equation}
  \rho_{k,l}^{i,j} = \lambda_{i+j} s_{n,m}(i,j;k,l) \delta_{i+j=k+l} + \sigma(\delta_{(i,j)=(m,n-m)}\delta_{(k,l)=(0,0)}+\delta_{(i,j)=(0,0)}\delta_{(k,l)=(m,n-m)}),
 \end{equation}
 where $\rho^{i,j}_{k,l} = \bra{i,j}\rho \ket{k,l}$ and
 \begin{equation}
 s_{n,m}(i,j;k,l)=\sqrt{{m \choose i}{n-m \choose j}{m \choose k}{n-m \choose l}}.
 \end{equation}
Hence, the partial transposition on $m$ qubits is
 \begin{equation}
(\rho^{\Gamma_m})^{i,j}_{k,l} = \lambda_{i+l} s_{n,m}(i,l;k,j)\delta_{j-i=l-k} + \sigma(\delta_{(i,j)=(0,n-m)}\delta_{(k,l)=(m,0)}+\delta_{(i,j)=(m,0)}\delta_{(k,l)=(0,n-m)}).
 \end{equation}
 The partially transposed state $\rho^{\Gamma_m}$ decomposes into a direct sum of different blocks, which we shall index by the parameter $n$. Note that there are $N+1$ solutions to the equation $n=j-i$ with $0 \leq i \leq m$ and $0 \leq j \leq N-m$. Hence, $n$ ranges from $-m$ to $N-m$. The effect of the coherences with strength $\sigma$ only interferes the $-m$-th and $(N-m)$-th blocks (which would be $1 \times 1$ if $\sigma = 0$) joining them into a $2\times 2$ block. This proves Eq. (\ref{eq:BlocksPT}).
 \end{proof}
 
\subsection{Proof of Lemma \ref{lem:Bnm}}
\label{proof:Bnm}
\begin{proof}
With the assumption that the $\lambda_m$ form a sequence satisfying the above recurrence, we can now easily span the kernel of $B_n^{(m)}$ for $-m+2\leq n \leq N-m-2$:
\begin{equation}
 (\begin{array}{ccccccccc}0 & \ldots & 0 & 1 & c_1& c_0 & 0 & \ldots & 0\end{array})^T \in \ker B_n^{(m)}.
\end{equation}

Hence, $B_n^{(m)}$ has rank at most $2$. Now we prove its rank is $2$ and $\rho$ is PPT. We do this by constructing a Cholesky decomposition (or a Gram decomposition) of $B_n^{(m)}$:

Let us define $I_{a,b}:=f_m\cdot f_{m+a+b} - f_{m+a} \cdot f_{m+b}$. Let us note that $I_{a,b}$ is independent of $m$:
$$I_{a,b} = \frac{(\alpha-1)(\beta-1)}{(\alpha-\beta)^2}\alpha^m \beta^m (\alpha^a - \beta^a)(\alpha^b - \beta^b),$$
but let us observe that $\alpha \beta = 1$. Hence,
$$I_{a,b} = \frac{(\alpha-1)(\beta-1)}{(\alpha-\beta)^2}(\alpha^a - \beta^a)(\alpha^b - \beta^b),$$
Since $Z\geq 0$ by definition, we have that $\alpha \geq \beta > 0$. Hence, $I_{a,b}\geq 0$. Furthermore, $I_{a,b}=\sqrt{I_{a,a} I_{b,b}}$.

It now follows that $B_n^{(m)}$ has rank $2$ whenever $m>0$ and it is positive semidefinite, admitting the following Cholesky decomposition $B_n^{(m)} = L_n^{(m)} \cdot (L_n^{(m)})^T$, where
\begin{equation}
(L_n^{(m)})^T = \frac{1}{\sqrt{f_p}}\left(
\begin{array}{cccccc}
f_{p} & f_{p+1} & f_{p+2} & f_{p+3} & \cdots & f_{p+q}\\
\sqrt{I_{0,0}} & \sqrt{I_{1,1}} & \sqrt{I_{2,2}} & \sqrt{I_{3,3}} & \cdots & \sqrt{I_{q,q}}
\end{array}
\right),
\end{equation}
with $p=2\max\{0,n\}-n$ and $q=2(\min\{m,N-m+n\}-\max\{0,n\})-n$.

We observe that the element on the $a$-th row, $b$-th column of $B_n^{(m)}$ corresponds to
\begin{eqnarray}
(B_n^{(m)})^a_b &=& \frac{1}{f_p}(f_{p+a}f_{p+b} - \sqrt{I_{a,a} I_{b,b}})=\frac{1}{f_p}(f_{p+a} f_{p+b} - I_{a,b})\nonumber\\
&=& \frac{1}{f_p}(f_{p+a}f_{p+k} + f_p f_{p+a+b} - f_{p+a} f_{p+b}) = f_{p+a+b}=f_{a+b-n}.
\end{eqnarray}
Observe that the property that $I_{a,b}$ is independent of $p$ becomes crucial.
\end{proof}

\subsection{Proof of Lemma \ref{lem:trace}}
\label{proof:trace}
\begin{proof}
To calculate the trace of $\rho$, it will be very useful to note the following identity: $f_{p} = f_{-p-1}$. Indeed, one can show that
$$f_p-f_{-p-1} = \frac{(1-\alpha)\alpha^{-1-p}+(\alpha-1)\alpha^p + (\beta-1) \beta^{-1-p} + (1-\beta)\beta^p}{\alpha-\beta}=0,$$
where the last identity easily follows from the property that $\alpha \beta = 1$.
Hence, the trace of $\rho$ reduces to the sum
\begin{equation}
\mathrm{Tr}(\rho) = \sum_{k=0}^{2 K + 1}{2K+1 \choose k} \lambda_k = \sum_{k=0}^{2 K + 1}{2K+1 \choose k} f_{K-k} 
\end{equation}
Let us note the following identity (it easily follows from Newton's binomial and $\beta=\alpha^{-1}$):
\begin{equation}
\sum_{k=0}^{2K+1}{2K+1 \choose k} \alpha^{K-k} = \alpha^K (1+\beta)^{2K+1}.
\end{equation}
This allows us to calculate
\begin{equation}
\mathrm{Tr}(\rho)=\sum_{k=0}^{2K+1}{2K+1 \choose k}\left(\frac{\alpha-1}{\alpha-\beta}\alpha^{K-k} - \frac{\beta-1}{\alpha-\beta}\beta^{K-k}\right) = \frac{\alpha-1}{\alpha-\beta}\alpha^K(1+\beta)^{2K+1} - \frac{\beta-1}{\alpha-\beta}\beta^K(1+\alpha)^{2K+1}.
\end{equation}
Since $\alpha \beta = 1$, we can express the trace of $\rho$ in terms of $\alpha$ and $\alpha^{-1}$:
\begin{equation}
\mathrm{Tr}(\rho) = (1+\alpha^{-1})^{2K}\alpha^K + (\alpha^{-1})^K(1+\alpha)^{2K}= 2\alpha^{-K}(1+\alpha)^{2K} = 2[(1+\alpha)(1+\alpha^{-1})]^K = 2[2+\alpha+\alpha^{-1}]^K.
\end{equation}
Hence, we arrive at the expression we wanted to prove:
\begin{equation}
\mathrm{Tr}(\rho) = 2(2+\alpha+\beta)^K = 2(4+Z)^K.
\end{equation}
\end{proof}

\end{document}